\documentclass{article}
\usepackage[utf8]{inputenc}

\usepackage[german,french,english]{babel}
\usepackage[T1]{fontenc}

\usepackage{fullpage}

\usepackage{amsmath}
\usepackage{amssymb}
\usepackage{amsthm}
\usepackage{stmaryrd}

\newtheoremstyle{itstyle}
  {1.5\topsep}{1.5\topsep}                                          
  {\itshape}{}{\bfseries}                               
  {.}{.5em}                                  
  {\thmname{#1} \thmnumber{#2}\thmnote{ (#3)}}  
\newtheoremstyle{notitstyle}
  {1.5\topsep}{1.5\topsep}                                          
  {}{}{\bfseries}                               
  {.}{.5em}                                  
  {\thmname{#1} \thmnumber{#2}\thmnote{ (#3)}}  

\theoremstyle{itstyle}
\newtheorem{theorem}{Theorem}[section]
\newtheorem{lemma}[theorem]{Lemma}
\newtheorem{proposition}[theorem]{Proposition}
\newtheorem{corollary}[theorem]{Corollary}

\theoremstyle{notitstyle}
\newtheorem{definition}[theorem]{Definition}

\newtheorem{assumption}[theorem]{Assumption}

\newtheorem{example}[theorem]{Example}

\usepackage{graphicx}

\usepackage{tikz}
\usetikzlibrary{decorations.pathmorphing,positioning,arrows,calc,chains,fit,shapes,automata,decorations.pathreplacing,calligraphy}

\usepackage{tikz-cd}
\tikzset{
    symbol/.style={
        draw=none,
        every to/.append style={
            edge node={node [sloped, allow upside down, auto=false]{$#1$}}}
    }
}
\tikzset{
modal/.style={
shorten >=1pt,
shorten <=1pt,
auto,
node distance=1.5cm,
semithick
},
world/.style={circle,draw,minimum size=0.5cm,fill=gray!15},
circ/.style={circle,draw,minimum size=0.5cm},
point/.style={circle,draw,inner sep=0.5mm,fill=black},
reflexive above/.style={->,loop,looseness=7,in=120,out=60},
reflexive below/.style={->,loop,looseness=7,in=240,out=300},
reflexive left/.style={->,loop,looseness=7,in=150,out=210},
reflexive right/.style={->,loop,looseness=7,in=30,out=330},
epi/.style={->>},
mono/.style={{Hooks[right]}->},
mono2/.style={{Hooks[left]}->},
rightmorph/.style={{Arc Barb}[reversed]-{Arc Barb}},
leftmorph/.style={-{Arc Barb}{Arc Barb}}
}

\newcommand\leftmorph[1][1.4em]{
\tikz[baseline=-0.7ex,shorten <=2pt, shorten >=2pt]
\draw[leftmorph] (0,0) -- (#1,0);
}
\newcommand\rightmorph[1][1.4em]{
\tikz[baseline=-0.7ex,shorten <=2pt, shorten >=2pt]
\draw[rightmorph] (0,0) -- (#1,0);
}

\usepackage{hyperref}
\hypersetup{
    colorlinks=true,
    linkcolor=blue,
    filecolor=magenta,      
    urlcolor=cyan,
}

\let\epsilon\varepsilon
\let\phi\varphi

\newcommand\slominski{S\l{}omi\'n{}ski}
\newcommand\adamek{Ad\'{a}mek}

\DeclareMathOperator*{\dom}{dom}

\title{Languages and Recognition in a Category with Factorisation}
\author{Harsh Beohar, Mike Cruchten and Georg Struth}
\date{\today}

\begin{document}

\maketitle

\begin{abstract}
  Language recognition by homomorphisms is a central construction of
  algebraic language theory. Initially studied for monoids and
  semigroups, it has subsequently been expanded to other algebraic
  structures.  Our new categorical account is based on fibrations,
  which have already seen other applications in automata theory.
  Languages and surjective homomorphisms give indeed rise to two
  fibrations, and the notion of language recognition is stable under
  reindexing.  We develop this framework in a category with a
  factorisation system and address two main technical questions in the
  fibrational setting. First, we provide sufficient conditions under
  which languages have syntactic quotients (which is a generalisation
  of syntactic congruences) and we show how such quotients can be
  described in some concrete cases using a result by
  \slominski{}. Second, we introduce sufficient conditions under which
  (regular) languages are closed under certain $\mathcal{J}$-limits
  and $\mathcal{J}$-colimits.
\end{abstract}

\section{Introduction}

Regular languages, their acceptance by finite automata and their
recognition by finite monoids are fundamental to computer
science. They have been studied from the point of view of
combinatorics, universal algebra, logic and increasingly category
theory, the language in which much of modern mathematics is written.
Categories allow describing automata and languages uniformly using
abstract concepts and deriving their properties from general
constructions.  An early use of categories in automata theory can be
found in Goguen's work~\cite{Goguen72:MachinesInCategory}. 
Automata internal to a category
were introduced by Arbib and Manes~\cite{arbib:1980:machines, arbib:1975:adjointMachines}
and further developed by Ad{\'a}mek and
Trnkov{\'a}~\cite{adamek:1990:automataAndAlgebrasInCategories}.
More recently, Rutten's coalgebraic perspective on automata 
\cite{rutten:2000:universalCoalgebra}  
has inspired extensive research (cf.\thinspace \cite{jacobs:2017:introductionToCoalgebra}).
Algebraic language theory, more specifically, studies the recognition
of languages by monoid homomorphisms.  Categorifications of this
approach include Boja\'nczyk's work on recognition by
monads~\cite{bojanczyk:2015:recognisableLanguagesOverMonads} and that
of Ad{\'a}mek, Milius and Urbat on recognition by monoid objects in symmetric
monoidal closed categories \cite{adamek:2018:catApproach}, for a few
paradigmatic examples.

Our motivation for introducing categorical concepts, and in particular
fibrations, to algebraic language theory is slightly different.  It is
standard to consider languages on the free monoid $A^\ast$ on an
alphabet $A$. It is then natural to consider the set of all languages
on $A^\ast$ as a fibre above $A^\ast$. Further, languages are usually
recognised by monoid homomorphisms from $A^\ast$, and such maps can
again be modelled as fibres above $A^\ast$. These two fibres can then
be combined into a single fibre above $A^\ast$, which is formed by
pairs of languages and monoid homomorphisms recognising the
former. Finally, this view can be extended to relate notions of
language and recognition with different alphabets.  It is, for
instance, known that regular languages are closed under inverse
images: if $h\colon B^\ast\to A^\ast$ is a monoid homomorphism and
$L\subseteq A^\ast$ is regular, then so is
$h^{-1}(L)\subseteq B^\ast$.  The proof of this result reveals that,
while the inverse image of a language is obviously again a language,
recognising maps can be reindexed along $h$ as well. From a
categorical point of view, it seems therefore natural to model this
situation in terms of a fibration.

In light of these observations, our main conceptual contribution is a
new categorical account of language recognition in terms of
Grothendieck fibrations.

Before describing our approach in more detail, we wish to point out
that he general idea of using fibrations in automata theory is not
new: Melliès and coworkers have used them in language
theory~\cite{mellies:2023:parsingAsLifting,mellies:2026:ccfibrationOfHigherOrderRegLang};
Barr and Wells describe their connection with wreath products in the
context of Krohn-Rhodes theory~\cite{barr:1990:categoryTheory}. Our
particular use of fibrations for defining algebraic language
recognition is, to our knowledge, novel.

Inspired by classical algebraic language theory, we replace monoids by
$T$-algebras for some monad $T$ on a category $\mathbb C$ and work
with a factorisation structure on $\mathbb C$. More specifically, we
use the left morphisms from the factorisation structure that are also
$T$-morphisms (in the Eilenberg-Moore category $\mathbb C^T$), which
we call \emph{quotients}, as recognisers. We construct a category
$\mathbb Q(T)$ of quotients (Definition~\ref{defn:QuotientCategory})
that can be arranged as a fibration
(Proposition~\ref{prop:qFibration}) whenever we assume that $T$
preserves left morphisms.

Furthermore, we define languages over an algebra simply as
characteristic functions into a fixed object $\Omega$ in $\mathbb
C$. This view gives rise to a `language' fibration $\mathbb L(T)$
(Definition~\ref{defn:LanguageFibration}) over the Eilenberg-Moore
category $\mathbb C^T$. Finally, we define a `recognition' fibration
$\mathbb R(T)$ (Definition~\ref{defn:RecFibration}) with respect to
the language and quotient fibrations, such that its fibres over any
algebra encode the recognition of a language by a quotient
(Definition~\ref{defn:languageRecognition}). This setup can be
summarised as in Figure \ref{eq:story}, in which each arrow in the diagram
corresponds to a fibration.

\begin{figure}[t]
  \centering
  \begin{tikzpicture}[modal]
    \node[] at (0,2) (a) [label=above:{}] {$\mathbb{R}(T)$};
    \node[] at (0,0) (b) [label=above:{}] {$\mathbb{Q}(T)$};
    \node[] at (3,2) (c) [label=above:{}] {$\mathbb{L}(T)$};
    \node[] at (3,0) (d) [label=above:{}] {$\mathbb{C}^T$};
    \path[->] (a) edge (b);
    \path[->] (a) edge (c);
    \path[->] (a) edge (d);
    \path[->] (b) edge (d);
    \path[->] (c) edge (d);
  \end{tikzpicture}
  \caption{The quotient, language and recognition fibrations.}
  \label{eq:story}
\end{figure}

Apart from these conceptual contributions, we study three technical
questions to which our fibrational framework seems most relevant.
\begin{itemize}
\item We introduce sufficient conditions for the existence of
  syntactic quotients (generalisations of syntactic congruences) in the
  presence of a factorisation system (Theorem
  \ref{thm:weakSyntacticObj}). We show that, when the category is equipped with
  a proper factorisation system, then these conditions follow
  from previous results by
  \adamek{}~\cite{adamek:1979:cogenerationOfAlgebras}.

\item We provide an explicit description of
  syntactic quotients via an interior operator (due to
  \slominski{}~\cite{slominski:1974:greatestCongruence}) for all
  concrete cases under consideration including algebras and ordered algebras
  over a finitary signature. We define the interior operator using the internal
  language of fibrations and \slominski{}'s elementary translations
  (Proposition~\ref{prop:slominski}).

\item We study regular languages relative to a distinguished class of
  finite objects.  We identify sufficient conditions under which
  regular languages form a fibration as well as conditions under which
  they are closed under certain limits and colimits (Theorem
  \ref{thm:limitsColimitsRegular}).
\end{itemize}

Our fibrational definition of recognition
(Definition \ref{defn:languageRecognition}) generalises the standard definition
of language recognition by monoid homomorphisms, with the exception that
we replace monoid homomorphisms by left $T$-morphisms for a given orthogonal
factorisation system and monad $T$. By instantiating our framework with
a trivial factorisation system we recapture the recognition by monoid
homomorphisms, while an instantiation to surjections and injections yields
the recognition by surjective monoid homomorphisms. Most of our examples
are worked out by considering the standard epi-mono factorisation over
the concrete categories under consideration.

We end this introduction with a discussion of related work.  The
existence of syntactic objects (e.g. syntactic monoids, minimal
automata and so forth) has been studied previously in different contexts.
Colcombet and Petri\c{s}an define automata as functors whose domain
category describes the automaton
structure~\cite{colcombet:2020:automataMinimization}. They discuss the
existence of minimal automata in the presence of initial and final
objects as well as a factorisation system, and they apply their techniques to
monoids in $\mathbf{Set}$ to describe the syntactic monoid for a
language on $A^\ast$.

\adamek{}, Milius and Urbat develop a categorical framework for
syntactic monoids in varieties of (ordered)
algebras~\cite{adamek:2018:catApproach}. These correspond to
Eilenberg-Moore categories of commutative, finitary monads over
$\mathbf{Set}$ or strongly finitary monads over $\mathbf{Pos}$ (which
are known to be symmetric monoidal closed), which leads the authors to
work abstractly in symmetric monoidal categories. Our approach, in
contrast, allows weaker assumptions, in particular non-commutative
monads. In addition, the existence of syntactic monoids in
\cite{adamek:2018:catApproach} relies on a stability condition and the
closure of epimorphisms under wide pushouts and tensor products. 

Bojanczyk's work on recognisable languages over monads features a
theorem on the existence of syntactic morphisms for a monad in the
sorted
setting~\cite{bojanczyk:2015:recognisableLanguagesOverMonads}. This
work uses polynomials over an algebra and is implicitly related to
(elementary) translations, as used in \slominski{}'s construction.
Translations have been widely used elsewhere for defining syntactic
congruences in algebraic language theory, see
\cite{almeida:1990:pseudovarieties,wilke:1993:algebraic} for early
references.  A non-categorical description of syntactic congruences in
the multisorted setting can be found in
\cite{vidal:2020:congruenceBasedProofs}.

Our work on the existence of syntactic quotients (Section
\ref{sec:syntactic}) is related to
\cite{adamek:2018:catApproach}. However, while \adamek{} et al.\
establish a result (\cite[Theorem 3.9]{adamek:2018:catApproach}) akin
to Proposition \ref{prop:quotJoins}, they do not obtain a right
adjoint to the forgetful functor, which we established in Corollary
\ref{cor:uHasRightAdj}. Our syntactic quotient result follows from
this right adjoint in the presence of the factorisation system (see
Theorem \ref{thm:weakSyntacticObj}). Bojanczyk's results in
\cite{bojanczyk:2015:recognisableLanguagesOverMonads} are closely
related to those in our section on \slominski{}'s construction
(Section \ref{sec:syntactic}).
While Bojanczyk's polynomials are more categorical than our treatment of
translations, it only applies to the multisorted setting,
whereas our work also covers the ordered setting.

\section{Preliminaries}

In this section we fix some notation, introduce factorisation systems
\cite{adamek:1990:joyOfCats} and recall some basic concepts related to
fibred categories \cite{jacobs:1999:categoricalLogic}.

We assume familiarity with basic category
theory~\cite{adamek:1990:joyOfCats,awodey:2010:categoryTheory,maclane98}.
Throughout this article we work with locally small categories. We
write $\mathbb{C}^T$ for the Eilenberg-Moore category of a monad
$(T,\eta,\mu)$ on the category $\mathbb{C}$, $\text{Alg}(G)$ for the
category of $G$-algebras of the functor $G:\mathbb{C}\to \mathbb{C}$,
and $\mathcal{A}$ for the $T$-algebra $(A,f)$ with carrier $A$ and
algebra map $f\colon TA\to A$. We write $\mathbb{C}\simeq \mathbb{D}$
if the categories $\mathbb{C}$ and $\mathbb{D}$ are (naturally)
equivalent.

\begin{definition}
  Let $\mathbb{C}$ be a category.
\begin{enumerate}
\item An \emph{(orthogonal) factorisation system} on $\mathbb{C}$ is a
  pair $(\mathcal{L},\mathcal{R})$ of a class of left
  morphisms $\mathcal{L}$ and a class of right morphisms $\mathcal{R}$
  such that
\begin{enumerate}
 \item $\mathcal{L}$ and $\mathcal{R}$ contain all isomorphisms and are closed under composition,
 \item  every morphism $f$ in $\mathbb{C}$ factorises as $f=m\circ e$ with $e\in \mathcal{L}$ and
      $m\in \mathcal{R}$,
    \item for each commuting square
 \begin{center}
    \begin{tikzpicture}[modal,scale=0.9]
      \node[] at (0,0) (a) [label=above:{}] {$A$};
      \node[] at (2,0) (b) [label=above:{}] {$B$};
      \node[] at (0,2) (c) [label=above:{}] {$C$};
      \node[] at (2,2) (d) [label=above:{}] {$D$};
      \path[rightmorph] (a) edge node[below]{$m$} (b);
      \path[->] (c) edge node[left]{$u$} (a);
      \path[leftmorph] (c) edge node[]{$e$} (d);
      \path[->] (d) edge node[]{$v$} (b);
      \path[->,dashed] (d) edge node[]{$w$} (a);
    \end{tikzpicture}
  \end{center}
with $e\in \mathcal{L}$ and $m\in \mathcal{R}$,
  there is a unique filler morphism $w$ for which both
  triangles commute.
\end{enumerate}
\item The factorisation system $(\mathcal{L},\mathcal{R})$ is \emph{proper} if
  $\mathcal{L} \subseteq \mathcal{E}$ and
  $\mathcal{R} \subseteq \mathcal{M}$, where $\mathcal{E}$ denotes the
  class of epimorphisms and $\mathcal{M}$ the class of monomorphisms
  in $\mathbb{C}$.
\end{enumerate}
\end{definition}

We indicate left and right morphisms using arrows $\leftmorph $ and
$\rightmorph $, respectively while we reserve $\twoheadrightarrow $ and
$\hookrightarrow $ for epimorphisms and monomorphisms, respectively.
We use $e$ (with indices) exclusively for left morphism, and $m,n$
(with indices) exclusively for right morphisms.

\begin{example}
  Denote by $\text{Mor}(\mathbb{C})$ and $\text{Iso}(\mathbb{C})$ the
  collection of all morphisms and isomorphism in a category
  $\mathbb{C}$ respectively. The pair
  $(\text{Mor}(\mathbb{C}),\text{Iso}(\mathbb{C}))$ is the
  \emph{trivial factorisation system on $\mathbb{C}$}. It is obvious that
  $\text{Mor}(\mathbb{C})$ and $\text{Iso}(\mathbb{C})$ are closed
  under composition and contain all isomorphisms. Moreover every
  morphism $f\colon A\to B$ can trivially be factored as
  $f = \text{id}_B\circ f$. For the filler property, let
  $v\circ e = m\circ u$ with $e$ a morphism and $m$ and isomorphism
  with inverse $m^{-1}$. Then the filler morphism
  is $m^{-1}\circ v$ for
 $m^{-1}\circ v \circ e = m^{-1}\circ m\circ u = u$ and $m\circ m^{-1}\circ v = v$.
  This factorisation system is in general not proper as the left morphisms
  may not be epimorphisms.
\end{example}

\begin{definition}
  Let $p\colon \mathbb{E}\to \mathbb{C}$ be a functor.
  \begin{enumerate}
  \item A morphism $f\colon X\to Y$ in $\mathbb{E}$ is
    \emph{Cartesian} over $u\colon I\to J$ in $\mathbb{C}$ if $p(f)=u$
    and, for any $g\colon Z\to Y$ and $w\colon p(Z)\to J$ such that
    $u\circ w = p(g)$, there exists a unique $h\colon Z\to X$ in
    $\mathbb{E}$ such that $w=p(h)$ and $f\circ h = g$. The morphism $f$ is
    \emph{Cartesian} if it is Cartesian over $p(f)$.

\item A functor $p$
is a \emph{(Grothendieck) fibration} if for every $Y$ in $\mathbb{E}$ and
$u\colon I\to p(Y)$ in $\mathbb{C}$ there is a Cartesian morphism
$f\colon X\to Y$ above $u$.

\item A \emph{cleavage} is a choice of Cartesian morphism
$\overline{u}(Y)\colon u^\ast(Y)\to Y$ above $u$, for every $Y$ and $u$.
A fibration is \emph{cloven} if it has a cleavage.
\end{enumerate}
\end{definition}

For any fibration $p\colon \mathbb{E}\to \mathbb{C}$, the \emph{fibre}
above an object $I$ of $\mathbb{C}$ is the category $\mathbb{E}_I$,
whose objects are those objects in $\mathbb{E}$ that map to $I$ under
$p$, and whose morphisms are those morphisms in $\mathbb{E}$ that map
to $\text{id}_I$ under $p$. If $p$ is
cloven and $u\colon I\to J$ in $\mathbb{C}$, then the \emph{reindexing
  map} $u^\ast\colon \mathbb{E}_J\to \mathbb{E}_I$ is a functor.

Each cloven fibration $p$ gives rise to a contravariant pseudofunctor
$\Phi\colon\mathbb{C}^{\text{op}}\to \mathbf{Cat}$ called
\emph{indexed category}, a `functor' into the bicategory
$\mathbf{Cat}$ for which associativity and identity laws only hold up
to natural isomorphisms. It sends $I$ to $\mathbb{E}_I$ and
$u\colon I\to J$ to $u^\ast$. 
We assume no additional knowledge of bicategories and
refer to \cite{johnson:2021:2DimensionalCategories} for more details.
Conversely, each indexed category $\Phi$
induces a fibration via the Grothendieck construction.  For
this, one forms the category $\int\Phi$, which has pairs $(I,X)$ in
$\mathbb{C} \times \Phi(I)$ as objects, while morphisms
$(I,X)\to (J,Y)$ are pairs $(u,f)$ such that $u\colon I\to J$ in
$\mathbb{C}$ and $f\colon X\to \Phi(u)(Y)$ in $\Phi(I)$.  The
projection $\int\Phi\to \mathbb{C}$ is then a cloven fibration.

A \emph{morphism of fibrations} from $p\colon \mathbb{E}\to \mathbb{C}$ to
$q\colon \mathbb{D}\to \mathbb{B}$ is a pair of functors $(K,H)$ with
$K\colon \mathbb{C}\to \mathbb{B}$ and $H\colon \mathbb{E}\to \mathbb{D}$
such that $K\circ p = q\circ H$ and such that $H$ maps Cartesian morphism in
$\mathbb{E}$ to Cartesian morphisms in $\mathbb{D}$. The functor $H$
is \emph{fibred}: it restricts to a functor $H_I\colon \mathbb{E}_I\to \mathbb{D}_{K(I)}$
for each object $I$ of $\mathbb{C}$.

Finally, a fibration $p\colon \mathbb{E}\to \mathbb{C}$ has
\emph{fibred $\mathcal{J}$-limits ($\mathcal{J}$-colimits)} if each fibre has $\mathcal{J}$-limits
($\mathcal{J}$-colimits) and they are preserved under reindexing.

\begin{example}
  Let $\mathbb{C}$ be any category. Recall that the \emph{arrow category}
  $\mathbb{C}^{\to}$ has morphisms $f\colon A\to B$ in
  $\mathbb{C}$ as objects and that a morphism from $f\colon A\to B$ to
  $g\colon C\to D$ in $\mathbb{C}^{\to}$ is a pair of morphisms $(h,i)$ (both in
  $\mathbb{C}$) such that $g\circ h = i\circ f$. The composition in
  $\mathbb{C}^{\to}$ is pointwise. The domain functor
  $\dom\colon \mathbb{C}^\to \to \mathbb{C}$ sends morphisms to their domains
  and a pair of morphisms to their first component. Let $f\colon J\to X$
  be an object in $\mathbb{C}^\to$ and $u\colon I\to J$ a morphism in $\mathbb{C}$.
  We show that $(u,\text{id}_X)\colon f\circ u \to f$ is a Cartesian
  morphism above $u$. It is clear that $(u,\text{id}_X)$ is a well-defined
  morphism in $\mathbb{C}^\to$ as $f\circ u = \text{id}_X\circ f\circ u$
  and it is above $u$ as $\dom(u,\text{id}_X)=u$. Next let $(v,i)\colon g\to f$
  in $\mathbb{C}^\to$ and $w\colon K\to I$ such that $u\circ w = v$ as in
  the diagram below. We have to show that there is a unique morphism
  $Y\to X$ (shown as a dashed arrow in the diagram), which makes
  everything commute.
  \begin{center}
  \begin{tikzpicture}[modal]
  \node[] at (0,0) (a) [label=above:{}] {$K$};
  \node[] at (0,2) (b) [label=above:{}] {$Y$};
  \node[] at (3,0) (c) [label=above:{}] {$I$};
  \node[] at (3,2) (d) [label=above:{}] {$X$};
  \node[] at (6,0) (e) [label=above:{}] {$J$};
  \node[] at (6,2) (f) [label=above:{}] {$X$};
  \path[->] (a) edge node[]{$g$}(b);
  \path[->] (a) edge node[]{$w$}(c);
  \path[->] (a) edge[bend right] node[below]{$v$} (e);
  \path[->,dashed] (b) edge node[]{}(d);
  \path[->] (b) edge[bend left] node[]{$i$} (f);
  \path[->] (c) edge node[]{$u$}(e);
  \path[->] (c) edge node[]{$f\circ u$}(d);
  \path[->] (d) edge node[below]{$\text{id}_X$}(f);
  \path[->] (e) edge node[right]{$f$}(f);
  \end{tikzpicture}
  \end{center}
  By choosing $i\colon Y\to X$ for the dashed arrow, everything commutes,
  hence $(u,\text{id}_X)\colon f\circ u\to f$ is Cartesian over $u$.
  Therefore, $\dom$ is a fibration, which we call \emph{domain fibration}.
  The fibre $\dom_J$ above $J$ is (isomorphic to) the coslice category
  $J/\mathbb{C}$. For each $f\colon J\to X$ and $u\colon I\to J$, we
  have a choice of Cartesian lifting with $\overline{u}(f)=(u,\text{id}_X)$
  and $u^\ast(f)=u\circ f$. This makes $\dom$ a cloven fibration.
\end{example}

\section{Quotients}\label{sec:quotients}

We now introduce quotients for a monad $(T,\eta,\mu)$ over a category
$\mathbb{C}$ with a factorisation system. We construct a category of
quotients equipped with a domain functor to the Eilenberg-Moore
category $\mathbb{C}^T$, which forms a fibration whenever $T$
preserves left morphisms. We instantiate this quotient fibration to
varieties of (ordered) algebras and to semiautomata in $\mathbf{Set}$.
Our results for monads apply equally to endofunctors $G$. The
Eilenberg-Moore category is then replaced by the category of
$G$-algebras.

\subsection{The Quotient Fibration}

For the remainder of this article we fix a category $\mathbb{C}$ with
a factorisation system and a monad $(T,\eta,\mu)$ on $\mathbb{C}$. Following
\cite{adamek:1979:cogenerationOfAlgebras}, a \emph{quotient} of a
$T$-algebra $\mathcal{A}$ is a left $T$-morphism
$e\colon \mathcal{A}\leftmorph \mathcal{B}$.

Quotients thus give us a notion of generalised epimorphism, which is
instrumental for language recognition below.  First we assemble quotients into a
category.

\begin{definition}\label{defn:QuotientCategory}
  Let $T$ be a monad. The category $\mathbb{Q}(T)$ of \emph{quotients
    over $\mathbb{C}^T$} has quotients (left $T$-morphisms)
  $e\colon \mathcal{A}\leftmorph \mathcal{B}$ as objects and pairs
  $(f,g)\colon e\to e'$ as morphisms, where $f$ in $\mathbb{C}^T$, $g$
  in $\mathbb{C}$ and $e'\circ f = g\circ e$.  Composition in
  $\mathbb{Q}(T)$ is defined component-wise; the identity morphism on
  $e:\mathcal{A}\leftmorph \mathcal{B}$ is the pair
  $(\text{id}_A,\text{id}_B)$.
\end{definition}

We write $\mathbb{Q}$ for $\mathbb{Q}(T)$ when $T$ is
clear from the context.  The  obvious \emph{domain functor}
$\dom_\mathbb{Q}\colon \mathbb{Q}\to \mathbb{C}^T$ maps a
quotient $e\colon \mathcal{A}\leftmorph \mathcal{B}$ to $\mathcal{A}$
and a morphism $(f,g)$ to $f$. The factorisation system on
$\mathbb{C}$ allows reindexing
$e\colon \mathcal{B} \leftmorph \mathcal{C}$ along
$h\colon \mathcal{A}\to \mathcal{B}$ by factoring $e\circ h$ as
$m\circ e'$ provided that $e'$ is again a
$T$-morphism. This is shown in the following diagram:
\begin{center}
  \begin{tikzpicture}[modal,scale=0.9]
    \node[] at (0,0) (a) [label=above:{}] {$D$};
    \node[] at (2,0) (b) [label=above:{}] {$C$};
    \node[] at (0,2) (c) [label=above:{}] {$A$};
    \node[] at (2,2) (d) [label=above:{}] {$B$};
    \path[rightmorph] (a) edge node[below]{$m$} (b);
    \path[leftmorph] (c) edge node[left]{$e'$} (a);
    \path[->] (c) edge node[]{$h$} (d);
    \path[leftmorph] (d) edge node[]{$e$} (b);
  \end{tikzpicture}
\end{center}

\begin{proposition}\label{prop:qFibration}
  The domain functor $\dom_\mathbb{Q}$ is a fibration whenever $T$
  preserves left morphisms:
  $T\mathcal{L}\subseteq \mathcal{L}$.
\end{proposition}

\begin{proof}
  Let $f\colon \mathcal{A}\to \mathcal{A}'$ and
  $e'\colon\mathcal{A}'\leftmorph \mathcal{B}'$. We need to show that $f$ has
  a Cartesian morphism with codomain $e'$ above it. We first take the
  factorisation of $e\circ f$ in $\mathbb{C}$, which leads to the following diagram
  (that means $e,m$ are morphisms in $\mathbb{C}$).
  \begin{center}
    \begin{tikzpicture}[modal]
      \node[] at (0,0) (a) [label=above:{}] {$B$};
      \node[] at (2,0) (b) [label=above:{}] {$\mathcal{B}'$};
      \node[] at (0,2) (c) [label=above:{}] {$\mathcal{A}$};
      \node[] at (2,2) (d) [label=above:{}] {$\mathcal{A}'$};
      \path[rightmorph] (a) edge node[below]{$m$} (b);
      \path[leftmorph] (c) edge node[left]{$e$} (a);
      \path[->] (c) edge node[]{$f$} (d);
      \path[leftmorph] (d) edge node[]{$e'$} (b);
    \end{tikzpicture}
  \end{center}
  It remains to show that $(f,m)$ is a Cartesian morphism above $f$. We
  proceed in two steps: first we show that $(f,m)$ is a well-defined
  morphism in $\mathbb{Q}(T)$ and second that it is Cartesian above $f$.
  
  First we show that $e$ is a $T$-morphism, so that $(f,m)$ is a morphism
  in $\mathbb{Q}(T)$. Let $a\colon TA\to A$ and
  $b\colon TB'\to B'$ be the algebra maps of $\mathcal{A}$ and
  $\mathcal{B}'$, respectively. By viewing the above square in the
  underlying category $\mathbb{C}$, we obtain the
  diagram on the left below. Each face in this partial cube commutes, and
  we need to show that there is a structure map $c\colon TB\to B$ for which
 the whole cube commutes. It is constructed using the filler
  property of the square on the right, using the fact that $Te$
  is a left morphism. From this square it also follows that $e$ and
  $m$ behave like $T$-morphisms in the sense that they are $T$-morphism if
  the filler map $c$ is a $T$-algebra.
  \begin{center}
    \begin{tikzpicture}[modal]
      \node[] at (0,0) (a) [label=above:{}] {$A$};
      \node[] at (2,0) (b) [label=above:{}] {$B$};
      \node[] at (1,1) (c) [label=above:{}] {$A'$};
      \node[] at (3,1) (d) [label=above:{}] {$B'$};
      \node[] at (0,2) (e) [label=above:{}] {$TA$};
      \node[] at (2,2) (f) [label=above:{}] {$TB$};
      \node[] at (1,3) (g) [label=above:{}] {$TA'$};
      \node[] at (3,3) (h) [label=above:{}] {$TB'$};
      \path[leftmorph] (a) edge node[below]{$e$} (b);
      \path[->] (a) edge node[]{$f$} (c);
      \path[rightmorph] (b) edge node[below right]{$m$} (d);
      \path[leftmorph] (c) edge node[]{$e'$} (d);
      \path[->] (e) edge node[left]{$a$} (a);
      \path[leftmorph] (e) edge node[]{} (f);
      \path[->] (e) edge node[]{} (g);
      \path[->] (f) edge node[]{} (h);
      \path[->] (g) edge node[]{} (c);
      \path[leftmorph] (g) edge node[]{} (h);
      \path[->] (h) edge node[]{$b$} (d);
      \node[] at (6,0.5) (a1) [label=above:{}] {$B$};
      \node[] at (8,0.5) (b1) [label=above:{}] {$B'$};
      \node[] at (6,2.5) (e1) [label=above:{}] {$TA$};
      \node[] at (8,2.5) (f1) [label=above:{}] {$TB$};
      \path[rightmorph] (a1) edge node[below]{$m$} (b1);
      \path[->] (e1) edge node[left]{$e\circ a$} (a1);
      \path[leftmorph] (e1) edge node[]{$Te$} (f1);
      \path[->] (f1) edge node[]{$b\circ T(m)$} (b1);
      \path[->,dashed] (f1) edge node[]{$c$} (a1);
    \end{tikzpicture}
  \end{center}
  As a structure map of a $T$ morphism, $c$ must interact properly
  with the unit and the multiplication of the monad. For the unit
  equation, we have to check that $c\circ \eta_{B}=\text{id}_{B}$, for
  which we consider the commuting square:
  \begin{center}
    \begin{tikzpicture}[modal]
      \node[] at (0,0) (a1) [label=above:{}] {$B$};
      \node[] at (2,0) (b1) [label=above:{}] {$B'$};
      \node[] at (0,2) (e1) [label=above:{}] {$A$};
      \node[] at (2,2) (f1) [label=above:{}] {$B$};
      \path[rightmorph] (a1) edge node[below]{$m$} (b1);
      \path[->] (e1) edge node[left]{$e$} (a1);
      \path[leftmorph] (e1) edge node[]{$e$} (f1);
      \path[->] (f1) edge node[]{$m$} (b1);
      \path[->,dashed] (f1) edge node[]{} (a1);
    \end{tikzpicture}
  \end{center}
  Clearly the unique filler is the identity map $\text{id}_{B}$. We claim that
  $c\circ \eta_{B}$ also makes both triangles commute. Using the fact that
  $m$ and $e$ behave like $T$-morphisms, the naturality of $\eta$ and the fact that
  $a$ and $b$ are $T$-algebras, 
  \begin{align*}
    c \circ \eta_{B}\circ e &= c \circ T(e)\circ \eta_A = e\circ a\circ \eta_A = e, \\
    m\circ c\circ\eta_{B} &= b\circ T(m)\circ \eta_{B}=b\circ \eta_{B'}\circ m = m.
  \end{align*}
  Hence $c\circ\eta_{B}=\text{id}_{B}$ follows from  uniqueness of the filler.

  For the multiplication, we must show that $c\circ T(c) = c \circ \mu_{B}$. We consider
  the commuting square 
  \begin{center}
    \begin{tikzpicture}[modal]
      \node[] at (0,0) (a1) [label=above:{}] {$B$};
      \node[] at (2,0) (b1) [label=above:{}] {$B'$};
      \node[] at (0,2) (e1) [label=above:{}] {$TTA$};
      \node[] at (2,2) (f1) [label=above:{}] {$TTB$};
      \path[rightmorph] (a1) edge node[below]{$m$} (b1);
      \path[->] (e1) edge node[left]{$e\circ a \circ T(a)$} (a1);
      \path[leftmorph] (e1) edge node[]{$TT(e)$} (f1);
      \path[->] (f1) edge node[]{$b \circ T(b) \circ TT(m)$} (b1);
      \path[->,dashed] (f1) edge node[]{} (a1);
    \end{tikzpicture}
  \end{center}
  and need to show that the two triangles commute.  Using the fact
  that all the squares in the cube diagram above commute,
  the naturality of $\mu$, and the fact that $a$ and $b$ are $T$-algebras we get
 \begin{align*}
        (c \circ\mu_{B})\circ TT(e) &= c\circ T(e)\circ \mu_A \\
        &= e\circ a\circ \mu_A \\
        &= e\circ a\circ T(a) \\
        &= c\circ T(e)\circ T(a) \\
        &= c\circ T(e\circ a) \\
        &= c \circ T(c \circ T(e)) \\
        &= (c\circ T(c)) \circ TT(e).
      \end{align*}
  and
      \begin{align*}
        m \circ (c\circ \mu_{B}) &= b\circ T(m)\circ \mu_{B} \\
        &= b\circ \mu_{B'}\circ TT(m) \\
        &= b \circ T(b)\circ TT(m) \\
        &= b\circ T(b \circ T(m)) \\
        &= b\circ T(m\circ c) \\
        &= b\circ T(m)\circ T(c) \\
        &= m\circ (c\circ T(c)).
      \end{align*}
  Hence $c\circ \mu^{B}=c\circ T(c)$, and $c$ is a $T$-algebra.

  Second, we need to check that $(f,m)$ is Cartesian. Let
  $e''\colon \mathcal{A}''\leftmorph \mathcal{B}''$, $(g,h):e''\to e'$ and
  $j:\mathcal{A}''\to \mathcal{A}$. We must show that there is a unique
  $k:\mathcal{B}''\to \mathcal{B}$ such that $(f,m)\circ (j,k)=(g,h)$ as
  in the diagram on the left below. This unique map is obtained via
  the filler property as shown in the diagram on the right.
  \begin{center}
    \begin{tikzpicture}[modal]
      \node[] at (0,2.5) (a) [label=above:{}] {$\mathcal{A}''$};
      \node[] at (0,0.5) (b) [label=above:{}] {$\mathcal{B}''$};
      \node[] at (1.5,2) (c) [label=above:{}] {$\mathcal{A}$};
      \node[] at (1.5,0) (d) [label=above:{}] {$\mathcal{B}$};
      \node[] at (3.5,2) (e) [label=above:{}] {$\mathcal{A}'$};
      \node[] at (3.5,0) (f) [label=above:{}] {$\mathcal{B}'$};
      \path[leftmorph] (a) edge node[left]{$e''$} (b);
      \path[->] (a) edge node[]{$j$} (c);
      \path[->] (a) edge[bend left] node[]{$g$} (e);
      \path[->,dashed] (b) edge node[below left]{$k$} (d);
      \path[->] (b) edge[bend left] node[above right]{$h$} (f);
      \path[leftmorph] (c) edge node[above left]{$e$} (d);
      \path[->] (c) edge node[]{$f$} (e);
      \path[rightmorph] (d) edge node[below]{$m$} (f);
      \path[leftmorph] (e) edge node[]{$e'$} (f);
      \node[] at (7,0.25) (a1) [label=above:{}] {$\mathcal{B}$};
      \node[] at (9,0.25) (b1) [label=above:{}] {$\mathcal{B}'$};
      \node[] at (7,2.25) (e1) [label=above:{}] {$\mathcal{A}''$};
      \node[] at (9,2.25) (f1) [label=above:{}] {$\mathcal{B}''$};
      \path[rightmorph] (a1) edge node[below]{$m$} (b1);
      \path[->] (e1) edge node[left]{$e\circ j$} (a1);
      \path[leftmorph] (e1) edge node[]{$e''$} (f1);
      \path[->] (f1) edge node[]{$h$} (b1);
      \path[->,dashed] (f1) edge node[]{$k$} (a1);
    \end{tikzpicture}
  \end{center}
  The fact that the right hand square commutes follows as
  $h\circ e'' = e'\circ g = e'\circ f\circ j = m\circ e\circ j$.
  Hence $\dom_\mathbb{Q}$ is a fibration.
\end{proof}

\begin{assumption}
  We henceforth assume that $T$ preserves left
  morphisms.
\end{assumption} 

We further assume that the quotient fibration is cloven, so that
we can choose a factorisation $e\circ h = m\circ h^\ast
e$. Note that the fibre $\mathbb{Q}_\mathcal{A}$ above $\mathcal{A}$ has
quotients $e\colon \mathcal{A}\leftmorph \mathcal{B}$ as objects; while a morphism
$h\colon e\to e'$ in $\mathbb{Q}_\mathcal{A}$ is a morphism
$h\colon B\to B'$ in $\mathbb{C}$ such that
$e'\circ h = e$.
In particular, when $T=\text{Id}_\mathbb{C}$, the quotient fibration reduces to
the \emph{cosubobject fibration} $\dom_{\text{CoSub}}\colon\text{CoSub}\to \mathbb{C}$, where
$\text{CoSub}=\mathbb{Q}(\text{Id}_\mathbb{C})$, whose objects are
cosubobjects, that is, 
left morphisms out of a given object in $\mathbb C$.

\subsection{Examples}

Next we instantiate the quotient fibration to monads on $\mathbf{Set}$ and
$\mathbf{Pos}$ and to an endofunctor $G$ on $\mathbf{Set}$ whose $G$-algebras
are deterministic semiautomata.

\paragraph*{Finitary Monads on $\mathbf{Set}$}\label{eg:finitaryMonads}

The connection between finitary monads and varieties of algebras allows us
to describe quotient fibres using notions from universal algebra
\cite{linton:1965:aspectsOfEquationalCategories}.
With the standard epi-mono factorisation system in $\mathbf{Set}$,
endofunctors preserve surjections, as surjections are split
epimorphisms.  It is known that the Eilenberg-Moore category
$\mathbf{Set}^T$ of a finitary monad $T$ on $\mathbf{Set}$ is
equivalent to a variety of algebras.
A quotient in $\mathbb{Q}_\mathcal{A}$ corresponds to a surjective homomorphism
$e\colon \mathcal{A}\twoheadrightarrow \mathcal{B}$ in the
variety of algebras.

Each such homomorphism induces a congruence on $\mathcal{A}$ via its
kernel.  Conversely, each congruence $\equiv$ on $\mathcal{A}$ induces
a canonical surjective homomorphism
$\phi\colon \mathcal{A}\twoheadrightarrow \mathcal{A}/{\equiv}$.  Both
constructions are functorial with respect to $\mathbb{Q}_\mathcal{A}$
and the complete lattice of congruences $\text{Cong}(\mathcal{A})$ as
a category with inclusions as morphisms.

\begin{proposition}\label{prop:quotCongEq}
 For every $T$-algebra
  $\mathcal{A}$, $\mathbb{Q}_\mathcal{A}\simeq \text{Cong}(\mathcal{A})$.
\end{proposition}

\begin{proof}
  Define the functor $K\colon \mathbb{Q}_\mathcal{A}\to \text{Cong}(\mathcal{A})$
  which sends a surjective homomorphism $e\colon \mathcal{A}\twoheadrightarrow \mathcal{B}$
  to the congruence $\ker e = \{(a,a')\in A\times A\mid e(a)=e(a')\}$.
  (For functoriality, it suffices to show that
  $e\leq e' \implies \ker e\subseteq \ker e'$, which is straightforward.)
It has a pseudo-inverse
  $C\colon \text{Cong}(\mathcal{A})\to \mathbb{Q}_\mathcal{A}$, which sends a
  congruence $\equiv$ on $\mathcal{A}$ to the canonical homomorphism
  $\phi\colon \mathcal{A}\to \mathcal{A}/{\equiv}$. The fact that
  $\mathbb{Q}_\mathcal{A}$ and $\text{Cong}(\mathcal{A})$
  form an equivalence follows from standard isomorphism theorems of
  universal algebra~\cite{burris:1981:universalAlgebra}.
\end{proof}

Recall that a category is \emph{bicomplete} if it is both complete and
cocomplete.

\begin{corollary}\label{cor:quotCongComplete}
  All quotient fibres of $T$-algebras are bicomplete.
\end{corollary}

\begin{proof}
  Let $\mathcal{A}$ be a $T$-algebra.
  As $\mathbb{Q}_\mathcal{A}$ is equivalent to $\text{Cong}(\mathcal{A})$, it
  forms a complete lattice \cite{burris:1981:universalAlgebra}.
\end{proof}

Meets in $\text{Cong}(\mathcal{A})$ are
intersections, while joins are transitive closures of
unions. The reindexing of a homomorphism
$h\colon\mathcal{A}\to \mathcal{B}$ is given by the inverse image map
$h^{-1}\colon \text{Cong}(\mathcal{B})\to \text{Cong}(\mathcal{A})$.

\begin{proposition}\label{prop:quotCongColimitsAndFibredLimits}
  The quotient fibration $\dom_\mathbb{Q}$ has fibred limits.
\end{proposition}

\begin{proof}
  The limits in $\mathbb{Q}_\mathcal{A}$ correspond to meets in $\text{Cong}(\mathcal{A})$.
  Meets in $\text{Cong}(\mathcal{A})$ are  intersections of sets,
  which are trivially preserved by inverse images (as every inverse
  image function over sets has a left adjoint - the direct image
  function).
\end{proof}

\begin{example}
  We give an example showing that colimits may not be preserved by reindexing.
  Consider the identity monad and look at the function
  $h\colon 2\to 1$. Note that a congruence in this case is simply an
  equivalence relation. Take the initial object $\Delta_1$ in $\text{EqRel}(1)$.
  Clearly  $h^{-1}(\Delta_1)=2\times 2$, which is
  different from the initial object $\Delta_2$. So $h^{-1}$ does not
  preserve arbitrary colimits.
\end{example}

Similar results hold for finitary monads on $\mathbf{Set}^k$ 
for $k\in \mathbb{N}$.  These correspond to varieties of multisorted algebras
over finitary signatures. An instance of interest is the variety of
Wilke algebras, which corresponds to the Eilenberg-Moore category
for a monad on $\mathbf{Set}^2$. We return to this example in
Section~\ref{sec:syntactic}. 

\paragraph*{Strongly Finitary Monads on $\mathbf{Pos}$}\label{eg:stronglyFinitaryMonads}

As in the previous example, we describe quotients using notions from
universal algebra.  Recall that $\mathbf{Pos}$ is the category of
posets and order-preserving maps.  Consider the factorisation system
given by order-preserving surjections and injections.  The
Eilenberg-Moore category $\mathbf{Pos}^T$ for a strongly finitary
monad $T$ on $\mathbf{Pos}$ is equivalent to a variety of ordered
algebras~\cite{adamek:2022:categoricalViewOrderedAlgebras}. One can
think of $T$ as an inequational theory over a signature with constants
and function symbols. Each such $T$ preserves order-preserving 
surjections, as it can be obtained by lifting a finitary monad on
$\mathbf{Set}$ \cite{adamek:2022:categoricalViewOrderedAlgebras}. The
quotient fibre $\mathbb{Q}_{(\mathcal{A},\leq_A)}$ has surjective
order-preserving homomorphisms
$h\colon (\mathcal{A},\leq_A)\twoheadrightarrow (\mathcal{B},\leq_B)$
as objects; a morphism between two such homomorphisms $h$ and
$g\colon (\mathcal{A},\leq _A)\twoheadrightarrow (\mathcal{C},\leq_C)$
is a (necessarily surjective) order-preserving homomorphism
$f\colon (\mathcal{B},\leq_B)\twoheadrightarrow
(\mathcal{C},\leq_{C})$ such that $g\circ f = h$.

In $\mathbf{Pos}$,
instead of a correspondence between quotients and congruences, there
is a correspondence with precongruences; see
\cite{bloom:1976:varietiesOfOrderedAlgebras} for details.

A \emph{precongruence} on the set $A$
is a preorder ${\lesssim}\supseteq {\leq_A}$ that is compatible
with the operations on $\mathcal{A}$.
Each order-preserving surjective homomorphism
$h\colon (\mathcal{A},\leq_A)\twoheadrightarrow (\mathcal{B},\leq_B)$
induces a preorder $\lesssim$ on $A$ via $a \lesssim a' \iff h(a)\leq_B h(a')$,
which is a precongruence on $(\mathcal{A},\leq_A)$. Conversely,
every such precongruence $\lesssim$ gives rise to a canonical homomorphism
$\phi\colon (\mathcal{A},\leq_A)\twoheadrightarrow (\mathcal{A}/{\approx},\lesssim)$
with ${\approx}={\lesssim}\cap {\lesssim}^\text{c}$.
Hence for every ordered algebra $(\mathcal{A},\leq)$, the category
$\mathbb{Q}_{(\mathcal{A},\leq)}$ is equivalent to the complete lattice
$\text{PreCong}(\mathcal{A},\leq)$ of precongruences on
$(\mathcal{A},\leq)$ (as a category with inclusions as
morphisms).

\begin{proposition}\label{prop:quotPreCongEq}
  For every ordered $T$-algebra $\mathcal{A}$,
  $\mathbb{Q}_\mathcal{A}\simeq \text{PreCong}(\mathcal{A})$.
\end{proposition}

\begin{proof}
  This follows from \cite[Proposition 1.3]{bloom:1976:varietiesOfOrderedAlgebras}.
\end{proof}

\begin{corollary}
  All quotient fibres of ordered $T$-algebras are bicomplete.
\end{corollary}

\begin{proof}
  Let $(\mathcal{A},\leq)$ be a ordered $T$-algebra.
  As $\mathbb{Q}_{(\mathcal{A},\leq )}$ is equivalent to the complete lattice
  $\text{PreCong}(\mathcal{A},\leq)$, it is bicomplete.
\end{proof}

Meets in $\text{PreCong}(\mathcal{A},\leq)$ are given by intersections,
joins by the transitive closure of union.
Reindexing in terms of lattices of precongruences is again
given by the inverse image map. Such maps preserve arbitrary meets,
so the following fact holds.

\begin{proposition}
  The quotient fibration $\dom_\mathbb{Q}$ has fibred limits.
\end{proposition}

\begin{proof}
  The argument follows along the same lines as that of Proposition
  \ref{prop:quotCongColimitsAndFibredLimits}.
\end{proof}

Particular instances of this example are ordered semigroups and
ordered monoids as discussed in
\cite{pin:1995:varietyTheoremWithoutComplementation}. 

\paragraph*{Deterministic Semiautomata}\label{eg:semiautomata}

We now fit automata into the fibrational approach. As
acceptance states are only needed for accepting languages, we drop
them for now and focus on semiautomata.

Let $(A^\ast,\cdot,\epsilon)$ be the free monoid on the set $A$.
A \emph{semiautomaton} over $A$ is a tuple
$(Q,\overline{q},\delta)$
with \emph{initial state} $\overline{q}\in Q$ and \emph{transition
  function} $\delta\colon Q\times A\to Q$, which extends from letters
to words as
$\delta\colon Q\times A^\ast \to Q$. Homomorphisms of semiautomata are
functions $f\colon Q\to Q'$ that preserve initial states and satisfy
$f(\delta(q,a))=\delta'(f(q),a)$.  Reassembling these data as
$\left[ \overline{q},\delta \right]\colon 1+Q\times A \to Q$ yields a
$G$-algebra for the endofunctor $GX=1+X\times A$ on $\mathbf{Set}$;
semiautomata homomorphisms become $G$-morphisms.

We consider the quotient fibre above the initial semiautomaton
$\mathcal{S}=(A^\ast,\epsilon,\sigma(u,a)=ua)$.

\begin{lemma}
  The function $\delta(\overline{q},-)\colon A^\ast\to Q$ is the
  unique semiautomaton morphism $\mathcal{S}\to (Q,\overline{q},\delta)$.
\end{lemma}

\begin{proof}
  The function $\delta(\overline{q},-)$ is a semiautomaton
  morphism by definition. For uniqueness, suppose
  $f\colon A^\ast\to Q$ satisfies
  $f(\epsilon)=\overline{q}$ and $f(ua)=\delta(f(u),a)$. Then
  $\delta(\overline{q},-)=f$ by structural induction
  on words: the base case holds by definition, while the induction step 
  shows that $
    \delta(\overline{q},ua)=\delta(\delta(\overline{q},u),a) =
    \delta(f(u),a)=f(ua)$. 
\end{proof}

A quotient in $\mathbb{Q}_\mathcal{S}$ corresponds to a semiautomaton
$(Q,\overline{q},\delta)$ for which the unique $\delta$ is surjective,
so that
for every $q\in Q$ there exists a word $u\in A^\ast$ such that
$\delta(\overline{q},u)=q$. The category $\mathbb{Q}_\mathcal{S}$
is therefore isomorphic to the full subcategory of reachable
semiautomata.

Moreover, it is known that for each semiautomaton
$(Q,\overline{q},\delta)$, the kernel of $\delta(\overline{q},-)$ is a
right congruence on $A^\ast$.  Conversely, every right congruence is
compatible with the transition function of $\mathcal{S}$ and hence
induces a canonical semiautomaton morphism
\cite{ginzburg:1968:algebraicTheoryOfAutomata}.  By analogy with
Propositions \ref{prop:quotCongEq} and $\ref{prop:quotPreCongEq}$, we
obtain the following result.

\begin{proposition}\label{prop:quotRCongEq}
  The quotient fibre $\mathbb{Q}_\mathcal{S}$ is isomorphic to the
  full subcategory of reachable semiautomata, and equivalent to the
  lattice of right congruences $\text{RCong}(A^\ast)$ on $A$.
\end{proposition}

\begin{proof}
  The fact that each semiautomaton gives rise to a right congruence on
  $A^\ast$ and vice versa is known
  \cite{ginzburg:1968:algebraicTheoryOfAutomata}.  This correspondence
  becomes an equivalence of categories when restricting to reachable
  semiautomata. The proof follows that of Proposition
  \ref{prop:quotCongEq}.
\end{proof}

As the lattice $\text{RCong}(A^\ast)$ is complete, we obtain the following corollary.

\begin{corollary}
  The category $\mathbb{Q}_\mathcal{S}$ is bicomplete.
\end{corollary}

\begin{proof}
  Similar to that of Corollary
  \ref{cor:quotCongComplete}.
\end{proof}

Semiautomata form an instance of a finitary monad on $\mathbf{Set}$ as
in the first example. 
As $G$ is a polynomial functor
on $\mathbf{Set}$, the forgetful functor
$U\colon \text{Alg}(G)\to \mathbf{Set}$ has a left adjoint (in which
case $G$ is called \emph{varietor} \cite{adamek:1990:automataAndAlgebras}).
This adjunction defines a monad $G^\ast$ -- the \emph{free monad on $G$}.
The signature of the equational theory for $G^\ast$ is induced by $G$; it
consists of one constant $c$ and one unary function symbol $a$ per $a\in A$.
As there are no equations, the free algebras are simply term algebras. The
initial semiautomaton $\mathcal{S}$ is the free $G^\ast$-algebra on
$\emptyset$, and a congruence on $\mathcal{S}$, regarded as a $G^\ast$-algebra,
is precisely a right congruence on $A^\ast$ 
(Proposition~\ref{prop:quotRCongEq}).
Moreover, adapting this example to $\mathbf{Pos}$ leads to ordered semiautomata
(see \cite{klima:2019:varietiesOfOrderedAutomata}) that are linked
to the previous example on strongly finitary monads on $\mathbf{Pos}$.

\section{The Language and Recognition Fibrations}

We now introduce a notion of language over a $T$-algebra and
investigate the recognition of languages via quotients in the
fibrational approach. To this end we introduce language and
recognition fibrations, among others.  We instantiate this categorical
framework using the concrete examples from Section
\ref{sec:quotients}.

A \emph{language} over a $T$-algebra $\mathcal{A}$ is a morphism
$p\colon A\to\Omega$ into a fixed object $\Omega$ in $\mathbb{C}$.
The fibre of languages above $\mathcal{A}$ should therefore be some
categorical structure on $\mathbb{C}(A,\Omega)$. Here we
only impose that such hom-sets form categories as stated below.

\begin{assumption}\label{ass:languagePseudofunctor}
  Assume that $\mathbb{C}(A,\Omega)$ forms a category for each
  $A\in \mathbb{C}$, and that each $T$-morphism
  $h\colon \mathcal{A}\to \mathcal{B}$ induces a functor 
  $h^\ast\colon \mathbb{C}(B,\Omega)\to \mathbb{C}(A,\Omega)$ such
  that $\Phi\colon (\mathbb{C}^T)^{\text{op}}\to \mathbf{Cat}$ defined
  by $\Phi(\mathcal{A})=\mathbb{C}(A,\Omega)$ and $\Phi h=h^\ast$ is a
  contravariant pseudofunctor.
\end{assumption}

\begin{definition}\label{defn:LanguageFibration}
  The \emph{category of languages}
  $\mathbb{L}(T)$ is the total category $\int\Phi$
  and the \emph{language fibration} is the fibration $L\colon \mathbb{L}(T)\to \mathbb{C}^T$
  obtained from the Grothendieck construction.
\end{definition}

By the Grothendieck construction, 
the category $\mathbb{L}(T)$ has
pairs $(\mathcal{A},p\colon A\to \Omega)$ with $\mathcal{A}$ in $\mathbb{C}^T$ 
and $p$ in $\mathbb{C}$ as objects, and morphisms
$(\mathcal{A},p)\to (\mathcal{B},q)$ are pairs $(h,f)$ such that
$h\colon \mathcal{A}\to \mathcal{B}$ in $\mathbb{C}^T$ and $f\colon p\to h^\ast q$
in $\mathbb{C}(A,\Omega)$. The functor $L\colon \mathbb{L}(T)\to \mathbb{C}^T$ 
is the projection to the first component.

We drop the letter $T$ from our notation whenever it is clear from
context.  Our setup of quotients and languages allows defining
language recognition as in algebraic language theory.

\begin{definition}\label{defn:languageRecognition}
  Let $T$ be a monad and $\mathcal{A}$ a $T$-algebra. A language
  $p\colon A\to \Omega$
  is \emph{recognised by a quotient} $e\colon \mathcal{A}\leftmorph \mathcal{B}$
  if and only if there is a language $q\colon B\to\Omega$ such that
  $p= q\circ e$.
\end{definition}

This definition allows combining language and quotient fibrations
into one. For this we first take the following pullback in $\mathbf{Cat}$.
\begin{center}
  \begin{tikzpicture}[modal,scale=0.9]
    \node[] at (0,0) (a1) [label=above:{}] {$\mathbb{Q}(T)$};
    \node[] at (4,0) (b1) [label=above:{}] {$\mathbb{C}^T$};
    \node[] at (0,2) (e1) [label=above:{}] {$\mathbb{Q}(T)\times_{\mathbb{C}^T}\mathbb{L}(T)$};
    \node[] at (4,2) (f1) [label=above:{}] {$\mathbb{L}(T)$};
    \path[->] (a1) edge node[below]{$\dom_\mathbb{Q}$} (b1);
    \path[->] (e1) edge node[left]{$\pi_1$} (a1);
    \path[->] (e1) edge node[]{$\pi_2$} (f1);
    \path[->] (f1) edge node[]{$L$} (b1);
  \end{tikzpicture}
\end{center}
\begin{definition}\label{defn:RecFibration}
  The category $\mathbb{R}(T)$ is the full subcategory of
  $\mathbb{Q}(T)\times_{\mathbb{C}^T}\mathbb{L}(T)$ with triples
  $(\mathcal{A},e,p)$ as objects, such that $e\colon \mathcal{A}\leftmorph \mathcal{B}$ 
  recognises $p\colon A\to\Omega$.
\end{definition}

The functors $\pi_1$ and $\pi_2$ are fibrations as
they are obtained by a change of base construction
\cite{jacobs:1999:categoricalLogic}. Moreover, 
$\dom_\mathbb{Q} \circ \pi_1$ and $L \circ \pi_2$ are fibrations as they
are compositions of fibrations
\cite{jacobs:1999:categoricalLogic}. 
Our aim is to turn $\mathbb{R}(T)\to \mathbb{C}^T$ into a
subfibration, which requires that recognition is stable under
reindexing. A sufficient condition is that reindexing in
the language fibration $\mathbb L(T)$ is obtained by composition, so that
$h^\ast p = p\circ h$.

\begin{lemma}\label{lem:fibredRecognition}
  Let $h\colon \mathcal{A}\to \mathcal{B}$ and let $e\in \mathbb{Q}(T)_\mathcal{B}$
  recognise $p\in \mathbb{L}(T)_\mathcal{B}$.
  If $h^\ast p = p\circ h$, then $h^\ast e$ recognises
  $h^\ast p$.
\end{lemma}

\begin{proof}
  As $e$ recognises $p$, there exists a morphism $p'\colon B\to \Omega$ such that
  $p'\circ e=p$. Let $m$ be the right morphism in the factorisation
  $m\circ f^\ast e = e\circ f$. We show that $p'\circ m\circ (f^\ast e) = f^\ast p$. 
  \begin{center}
    \begin{tikzpicture}[modal]
      \node[] at (0,0) (a) [label=above:{}] {$f^\ast\mathcal{B}'$};
      \node[] at (3,0) (b) [label=above:{}] {$\Omega$};
      \node[] at (6,0) (c) [label=above:{}] {$\mathcal{B}'$};
      \node[] at (1.5,1.5) (d) [label=above:{}] {$\mathcal{A}$};
      \node[] at (4.5,1.5) (e) [label=above:{}] {$\mathcal{B}$};
      \path[->,dashed] (a) edge node[below]{$p'\circ m$} (b);
      \path[rightmorph] (a) edge[bend right] node[below]{$m$} (c);
      \path[->] (c) edge node[]{$p'$} (b);
      \path[leftmorph] (d) edge node[above left]{$f^\ast e$} (a);
      \path[->] (d) edge node[left]{$f^\ast p$} (b);
      \path[->] (d) edge node[]{$f$} (e);
      \path[->] (e) edge node[right]{$p$} (b);
      \path[leftmorph] (e) edge node[]{$e$} (c);
    \end{tikzpicture}
  \end{center}
  Indeed, as in the above diagram, we have that
  $p'\circ m\circ (f^\ast e) = p'\circ e\circ f = p\circ f = f^\ast p$.
  Hence $f^\ast e$ recognises $f^\ast p$ by definition.
\end{proof}

\begin{assumption}\label{ass:languageReindexing}
  Henceforth, we assume that reindexing in $\mathbb L(T)$ is obtained by composition:
  for any $h\colon \mathcal{A}\to \mathcal{B}$ and $p\in \mathbb{L}_\mathcal{B}$,
 the reindexing of $p$ along $h$ is given by
  $h^\ast p = p\circ h$.
\end{assumption}

We obtain the functor
$\dom_\mathbb{R}\colon \mathbb{R}(T)\to \mathbb{C}^T$ by composing
$\dom_\mathbb{Q}\circ \pi_1\circ I$ ( $=L \circ \pi_2\circ I$), which
results in the following diagram.  Moreover, by forgetting the
pullback object, we recover the commutative square from the introduction
(Figure \ref{eq:story}).

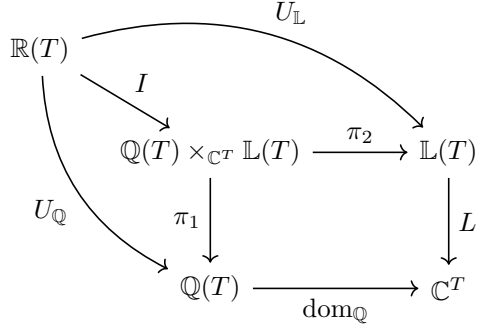
\begin{figure}[t]
  \centering
  \begin{tikzpicture}[modal,scale=0.9]
    \node[] at (2.5,0) (a1) [label=above:{}] {$\mathbb{Q}(T)$};
    \node[] at (6,0) (b1) [label=above:{}] {$\mathbb{C}^T$};
    \node[] at (2.5,2) (e1) [label=above:{}] {$\mathbb{Q}(T)\times_{\mathbb{C}^T}\mathbb{L}(T)$};
    \node[] at (6,2) (f1) [label=above:{}] {$\mathbb{L}(T)$};
    \node[] at (0,3.5) (g) {$\mathbb{R}(T)$};
    \path[->] (a1) edge node[below]{$\dom_\mathbb{Q}$} (b1);
    \path[->] (e1) edge node[left]{$\pi_1$} (a1);
    \path[->] (e1) edge node[]{$\pi_2$} (f1);
    \path[->] (f1) edge node[]{$L$} (b1);
    \path[->] (g) edge[bend right] node[below left]{$U_\mathbb{Q}$} (a1);
    \path[->] (g) edge[bend left] node[]{$U_\mathbb{L}$} (f1);
    \path[->] (g) edge[] node[]{$I$} (e1);
  \end{tikzpicture}
  \caption{The construction of the recognition fibration.}
  \label{eq:derivedstory}
\end{figure}

\begin{proposition}
  The domain functor $\dom_\mathbb{R}$ is a cloven fibration.
\end{proposition}

\begin{proof}
  Let $(\mathcal{B},e,p)$ be a triple in $\mathbb{R}$ and hence also
  in $\mathbb{Q}\times_{\mathbb{C}^T}\mathbb{L}$, let
  $f\colon \mathcal{A}\to \mathcal{B}$ in $\mathbb{C}^T$. The
  reindexing in $\mathbb{Q}\times_{\mathbb{C}^T}\mathbb{L}$ is
  performed pointwise, that is, $f^\ast(e,p)=(f^\ast e,f^\ast p)$, but
  this Cartesian morphism also exists in $\mathbb{R}$ due to Lemma
  \ref{lem:fibredRecognition}.  Hence
  $\dom_\mathbb{R}\colon \mathbb{R}\to \mathbb{C}^T$ is a subfibration
  of $\mathbb{Q}\times_{\mathbb{C}^T}\mathbb{L}\to
  \mathbb{C}^T$. It is clear from the construction that this
  fibration is cloven.
\end{proof}

\begin{definition}
  The fibration $\dom_\mathbb{R}\colon \mathbb{R}(T)\to \mathbb{C}^T$ is
  called the \emph{recognition fibration}.
\end{definition}

We write $U_\mathbb{Q}$ and $U_\mathbb{L}$ for the functors
$\pi_1\circ I$ and $\pi_2\circ I$, respectively, which are
fibrations as well. Moreover, as $\dom_\mathbb{Q}$ and $L$ are both
cloven, so are the fibrations $\dom_\mathbb{R}$, $U_\mathbb{Q}$ and
$U_\mathbb{L}$.
Whenever it is clear from the context, we write $\mathbb{R}$ for
$\mathbb{R}(T)$.

In sum, we have now introduced the five fibrations shown in
Figure \ref{eq:derivedstory}. First we constructed
the quotient fibration
$\dom_\mathbb{Q}\colon \mathbb{Q}\to \mathbb{C}^T$, whose fibre
$\mathbb{Q}_\mathcal{A}$ above $\mathcal{A}$ consists of all quotients
$e\colon \mathcal{A} \leftmorph \mathcal{B}$. Next we introduced
languages and the language fibration
$L\colon \mathbb{L}\to \mathbb{C}^T$. The language fibre
$\mathbb{L}_\mathcal{A}$ above $\mathcal{A}$ consists of all languages
$p\colon A\to \Omega$. From these fibrations we constructed three more
fibrations. The recognition fibration
$\dom_\mathbb{R}\colon \mathbb{R}\to \mathbb{C}^T$, whose fibre
$\mathbb{R}_\mathcal{A}$ above $\mathcal{A}$ consists of all pairs
$(e,p)$ such that $e$ recognises $p$. For the fibration
$U_\mathbb{Q}$, the fibre $\mathbb{R}_{(\mathcal{A},e)}$ above a
quotient $(\mathcal{A},e)$ consists of the languages
$p\colon A\to \Omega$ recognised by $e$. Finally, for the fibration
$U_\mathbb{L}$, the fibre $\mathbb{R}_{(\mathcal{A},p)}$ above the
language $(\mathcal{A},p)$ consists of  the quotients
$e\colon \mathcal{A}\leftmorph \mathcal{B}$ recognising $p$.

Next we revisit the examples from Section \ref{sec:quotients} and provide
instantiations of the language fibration in this context.

\begin{example}\label{eg:finitaryMonadsLanguage} 
  In $\mathbf{Set}$ we fix $\Omega = 2$. A language over
  $\mathcal{A}$ is a subset $P\subseteq A$ represented by
  $\chi_P\colon A\to 2$. For each $\mathcal{A}$, the fibre $\mathbb{L}_\mathcal{A}$
  above $\mathcal{A}$ is the hom-set $\mathbf{Set}(A,2)$. It forms 
  a complete atomic Boolean algebra  with inclusions as morphisms.
  Reindexing of $P\subseteq B$ along a homomorphism
  $h\colon\mathcal{A}\to \mathcal{B}$ means taking the
  inverse image $h^{-1}(P)\subseteq A$. As this is equivalent to
  $h^\ast\chi_P = \chi_P\circ h$, the reindexing satisfies Assumption
  \ref{ass:languageReindexing}. It is routine to check that
  reindexing preserves all meets and joins, so that the language
  fibration has all fibred limits and colimits.

  Let $e\colon \mathcal{A}\twoheadrightarrow \mathcal{B}$ and $P\subseteq A$.
  Then $e$ recognises $P$ if and only if there exists
  a subset $Q\subseteq B$ such that $\chi_Q \circ e = \chi_P$ or equivalently
  $P=e^{-1}[Q]$.

  By Proposition \ref{prop:quotCongEq},
  $\mathbb{Q}_\mathcal{A}\simeq\text{Cong}(\mathcal{A})$, so
  recognition should be definable in terms of congruences. Indeed,
  recognising a language means saturating a subset with an equivalence
  relation. Formally,
  an equivalence $\sim$ on
  $A$ \emph{saturates} a subset $P\subseteq A$ if and only if
  $P=\bigcup_{a\in P}[a]_\sim$. 
  Equivalently,
  $\sim$ saturates $P$ if and only if $P=\phi^{-1}[\phi(P)]$ where
  $\phi\colon A\to A/\sim$ is the canonical function associated with
  $\sim$.  It is routine to verify that a quotient
  $e\colon \mathcal{A}\twoheadrightarrow \mathcal{B}$ recognises a
  language $P\subseteq A$ if and only if $\ker e$ saturates
  $P$.
\end{example}

\begin{example}\label{eg:strongFinitaryMonadsLanguage}
  In $\mathbf{Pos}$ we fix $\Omega = (2,0\leq_2 1)$.
  A language over $(\mathcal{A},\leq)$ is an order-preserving map
  $(A,\leq)\to (2,\leq_2)$.  Such maps correspond
  precisely to subsets of $A$ that are up-closed with respect to
  $\leq$. The set of such subsets forms a complete lattice, which need
  not be complemented. Reindexing is based on inverse
  images, so Assumption \ref{ass:languageReindexing} is satisfied.  As
  reindexing preserves all meets and joins, the language fibration has
  fibred limits and colimits. 

  The definition of language acceptance instantiates as in Example
  \ref{eg:finitaryMonadsLanguage}. The definition can also be stated
  in terms of saturation: a precongruence $\lesssim$ on
  $(\mathcal{A},\leq)$ recognises an up-closed set
  $P\subseteq A$ if and only if the induced equivalence relation
  ${\sim}$ saturates $P$.
\end{example}

\begin{example}\label{eg:semiautomataLanguage}
  For semiautomata, a language over the initial semiautomaton
  $\mathcal{S}$ is a subset of $A^\ast$.  The language $L$ is accepted
  by
  $\delta(\overline{q},-)\colon \mathcal{S}\twoheadrightarrow
  (Q,\overline{q},\delta)$ if and only if there is a subset
  $F\subseteq Q$ of accepting states such that
  $ u\in L \iff \delta(\overline{q},u)\in F $. Hence our definition of
  language recognition specialises to the standard definition of
  language acceptance by an automaton.  Equivalently, recognition can
  again be defined in terms of saturation, the details of which we
  omit. 
\end{example}

The above examples demonstrate how classical concepts from automata and algebraic 
language theory are instantiated uniformly within our fibrational
framework, as summarised in Figure \ref{fig:table}.
Other instances that fall within our
setup include ordered automata, Wilke algebras and lasso
automata
\cite{klima:2019:varietiesOfOrderedAutomata,wilke:1993:algebraic,ciancia:2019:omegaAutomata}.

\begin{figure}[t]
  \centering
  \begin{tabular}{ |c|c|c|c|c|c | } 
   \hline
  $\mathbb C$ & $T\colon \mathbb C \to \mathbb C$ & $\mathbb Q_\mathcal A$ & $\mathbb L_\mathcal A$ & $\mathbb R_{(\mathcal{A},e)}$ \\ 
  \hline
  $\mathbf{Set}$ & Finitary monad  & Congruences & Subsets & Recognised languages\\
  $\mathbf{Pos}$ & \begin{tabular}[x]{@{}c@{}}Strongly\\ finitary monad \end{tabular}  & Precongruences & Up-closed subsets & Recognised languages\\
  $\mathbf{Set}$ & $1 + A \times -$ & Reachable semiautomata & Subsets & Accepted languages  \\
   \hline
  \end{tabular}
  \caption{Summary of examples and their fibrations.}
  \label{fig:table}
\end{figure}

\section{Syntactic Quotients}\label{sec:syntactic}

We now present sufficient conditions for the existence of syntactic
quotients. Intuitively, in accordance with classical language theory,
the syntactic quotient of a language should be the `smallest' quotient
recognising it, in the sense that it factors through any other quotient
recognising the language. It plays a role
analogous to the syntactic monoid (syntactic congruence) in algebraic
language theory, and to the minimal (semi)automaton in automata
theory.

We begin by defining syntactic quotients. Then we give sufficient
conditions under which a syntactic quotient exists and present some
relevant results by \adamek~\cite{adamek:1979:cogenerationOfAlgebras} which
establish these sufficient conditions for categories with proper factorisation
systems. Finally, we present a
construction of syntactic quotients using translations that originates
from results by \slominski{}
\cite{slominski:1974:greatestCongruence}. We conclude by showing how
\slominski's formalism can be expressed using fibrations, and how it can be
used to derive several syntactic congruences from the literature.

\begin{definition}
  Let $\mathcal{A}$ be a $T$-algebra. The \emph{syntactic quotient} for the
  language $p\colon A\to \Omega$ is the terminal object in
  $\mathbb{R}_{(\mathcal{A},p)}$.
\end{definition}

An equivalent, arguably more operational definition is that $e$ is the
syntactic quotient for $p$ if and only if it recognises $p$ and
whenever some $e'\in \mathbb{Q}_\mathcal{A}$ recognises $p$, there
exists a unique map $f:e'\to e$ in $\mathbb{Q}(T)_\mathcal{A}$.  If
$e$ is weakly terminal in
$\mathbb{R}_{(\mathcal{A},p)}$ (that is $f$ need not be unique), then we call $e$ a
\emph{weakly syntactic quotient} for $p$.

\subsection{Existence of Syntactic Quotients}

We now present sufficient conditions under which  languages admit
 (weakly) syntactic quotients. Then we outline existing results
obtained by \adamek~\cite{adamek:1979:cogenerationOfAlgebras} that allow us to
establish these conditions.

Our existence result is based on two technical lemmas. We can lift
the forgetful functor $U\colon \mathbb{C}^T\to \mathbb{C}$ for the monad
$T$ on $\mathbb{C}$ to a functor
$\overline{U}\colon \mathbb{Q}(T)\to \text{CoSub}$.  It maps a
quotient $e\colon \mathcal{A}\leftmorph \mathcal{B}$ to itself (seen
as a morphism in $\mathbb{C}$), and a pair $(f,g)$ again to itself. So
$\overline{U}$ just forgets that our maps are $T$-morphisms.

\begin{lemma}
  The pair $(U,\overline{U})$ is a morphism of fibrations; the following
  square commutes:
  \begin{center}
    \begin{tikzpicture}[modal,scale=0.9]
      \node[] at (0,0) (a) [label=above:{}] {$\mathbb{C}^T$};
      \node[] at (3,0) (b) [label=above:{}] {$\mathbb{C}$};
      \node[] at (0,2) (c) [label=above:{}] {$\mathbb{Q}(T)$};
      \node[] at (3,2) (d) [label=above:{}] {$\text{CoSub}$};
      \path[->] (a) edge node[below]{$U$} (b);
      \path[->] (c) edge node[left]{$\dom$} (a);
      \path[->] (c) edge node[]{$\overline{U}$} (d);
      \path[->] (d) edge node{$\dom$} (b);
    \end{tikzpicture}
  \end{center}
\end{lemma}

\begin{proof}
  We need to check that $\overline{U}$ preserves Cartesian morphisms,
  which is trivial because the Cartesian morphism were computed in
  $\mathbb{C}$ in the first place.
\end{proof}

Thus $\overline{U}$ is fibred and for
every $T$-algebra $\mathcal{A}$ there is a functor
$\overline{U}_\mathcal{A}\colon \mathbb{Q}_\mathcal{A}\to \text{CoSub}_A$.

\begin{lemma}\label{prop:coarsestLeftMorph}
  Every language $p\colon A\to\Omega$ has a weakly syntactic quotient
  in $\text{CoSub}_A$. It has a
  syntactic quotient if $\mathcal{L}\subseteq \mathcal{E}$.
\end{lemma}

\begin{proof}
  The statement is equivalent to saying that
  the fibre $\mathbb{R}(\text{Id}_\mathbb{C})_{(A,p)}$ has
  a weakly terminal object.
  Let $A\overset{e}{\to}B \overset{m}{\to}\Omega$ be the factorisation of $p$.
  We show that $e$ is weakly terminal.
  Assume that $e'\colon A\to C$ recognises $p$. Then there is a $p'$ such that
  $p=p'\circ e'$, as in the following diagram on the left below.
  \begin{center}
    \begin{tikzpicture}[modal]
      \node[] at (0,3) (a) [label=above:{}] {$A$};
      \node[] at (3,3) (b) [label=above:{}] {$C$};
      \node[] at (1.5,1.5) (c) [label=above:{}] {$B$};
      \node[] at (1.5,0) (d) [label=above:{}] {$\Omega$};
      \path[leftmorph] (a) edge node[]{$e'$} (b);
      \path[leftmorph] (a) edge node[]{$e$} (c);
      \path[->] (a) edge[bend right] node[below left]{$p$} (d);
      \path[->] (b) edge[bend left] node[]{$p'$} (d);
      \path[rightmorph] (c) edge node[]{$m$} (d);
      \node[] at (6,0.5) (a1) [label=above:{}] {$B$};
      \node[] at (8,0.5) (b1) [label=above:{}] {$\Omega$};
      \node[] at (6,2.5) (e1) [label=above:{}] {$A$};
      \node[] at (8,2.5) (f1) [label=above:{}] {$C$};
      \path[rightmorph] (a1) edge node[below]{$m$} (b1);
      \path[leftmorph] (e1) edge node[left]{$e$} (a1);
      \path[leftmorph] (e1) edge node[]{$e'$} (f1);
      \path[->] (f1) edge node[]{$p'$} (b1);
      \path[->,dashed] (f1) edge node[]{} (a1);
    \end{tikzpicture}
  \end{center}
  We obtain the commuting square shown on the right, and by the filler
  property a map from $e'$ to $e$, so that $e$ is weakly terminal.

  Finally, if $\mathcal{L}\subseteq \mathcal{E}$, then the fibre
  $\text{CoSub}_A$ is thin and hence $e$ is terminal.
\end{proof}

The (weakly) syntactic quotient $e$ in $\text{CoSub}_A$ is obtained by
factoring $p=m\circ e$, which requires a factorisation
system. The next theorem gives sufficient conditions for the existence
of (weakly) syntactic quotients.  If $\overline{U}_\mathcal{A}$ has a
right adjoint $\overline{R}_\mathcal{A}$, then
$\overline{R}_\mathcal{A}(e)$ is the (weakly) syntactic quotient for
$p$ in $\mathbb{Q}_\mathcal{A}$. A more descriptive explanation of
this result is given using \slominski{}'s construction below.

\begin{theorem}\label{thm:weakSyntacticObj}
  Let $\mathcal{A}$ be a $T$-algebra and let
  $\overline{U}_\mathcal{A}$ have
  a right adjoint. Then every language $p\colon A\to\Omega$ has
  a weakly syntactic quotient in $\mathbb{Q}(T)_\mathcal{A}$,
  and a syntactic quotient if $\mathcal{L}\subseteq \mathcal{E}$.
\end{theorem}

\begin{proof}
  We show that $\mathbb{R}(T)_{(\mathcal{A},p)}$ has a (weakly) terminal object.
  Let $\overline{R}_\mathcal{A}$ be right adjoint to $\overline{U}_\mathcal{A}$.
  We show that this adjunction restricts to an adjunction
  \begin{center}
  \begin{tikzpicture}[modal]
    \node[] at (0,0) (a) [label=above:{}] {$\mathbb{R}(T)_{(\mathcal{A},p)}$};
  \node[] at (3,0) (b) [label=above:{}] {$\mathbb{R}(\text{Id}_{\mathbb{C}})_{(A,p)}.$};
  \node[] at (1.5,0) (c) [label=above:{}] {$\bot$};
  \path[->] (a) edge[bend left] node[]{$\overline{U}_\mathcal{A}$} (b);
  \path[->] (b) edge[bend left] node[]{$\overline{R}_\mathcal{A}$} (a);
  \end{tikzpicture}
  \end{center}

  It is clear that $\overline{U}_\mathcal{A}$ restricts appropriately.
  For $\overline{R}_\mathcal{A}$ we need to show that if $e\in \text{CoSub}_A$
  recognises $p$, then so does $\overline{R}_\mathcal{A}(e)$. This follows
  directly from \ref{prop:recMDownClosed} 
  and the fact that
  $\epsilon_e\colon \overline{R}_\mathcal{A}(e)\to e$.

  Moreover, $\mathbb{R}(\text{Id}_\mathbb{C})_{(A,p)}$ has a weakly
  terminal object $e$ by Proposition \ref{prop:coarsestLeftMorph} . As
  $\overline{R}_\mathcal{A}$ is a right adjoint, it preserves weakly
  terminal objects and thus $\overline{R}_\mathcal{A}(e)$ is weakly
  terminal in $\mathbb{R}(T)_{(\mathcal{A},p)}$.

  Finally, if $\mathcal{L}\subseteq \mathcal{E}$, then $e$ is
  terminal, hence so is $\overline{R}_\mathcal{A}(e)$.
\end{proof}

Next we prove  three facts that also require a proper
factorisation system.  In this case it turns out that the quotient
fibration is \emph{thin}; that is,  each fibre $\mathbb{Q}_\mathcal{A}$ is a poset.  We write $h\colon e\leq e'$ whenever $h$ is the unique
$T$-morphism for which $e'\circ h = e$.  Moreover, instead of limits
and colimits, we speak of meets and joins in this
context.

The first two results are translated from
\cite{adamek:1979:cogenerationOfAlgebras}.  The first lemma gives
sufficient conditions for each fibre of the cosubobject fibration to
have arbitrary joins. As pointed out in
\cite{adamek:1979:cogenerationOfAlgebras}, the assumptions made are
weak enough to admit many interesting examples. Recall that a category
is co-well-powered if every object has a small poset of
cosubobjects.

\begin{lemma}[\cite{adamek:1979:cogenerationOfAlgebras}]\label{lem:coSubJoins}
  Let $\mathbb{C}$ be a co-well-powered category, which is either complete or
  cocomplete. Then each fibre in $\dom_\text{CoSub}$ has arbitrary
  joins.
\end{lemma}

The next result is the main theorem in
\cite{adamek:1979:cogenerationOfAlgebras}. Here we provide an alternative
fibrational proof.

\begin{proposition}[\cite{adamek:1979:cogenerationOfAlgebras}]\label{prop:quotJoins}
  The functor $\overline{U}_\mathcal{A}$ creates all joins in $\text{CoSub}_A$
  that are preserved by $T$.
\end{proposition}

\begin{proof}
  Let $\mathcal{A}=(A,a\colon TA\to A)$, $(e_i\colon \mathcal{A}\twoheadrightarrow \mathcal{A}_i)_{i\in I}$
  and $e\colon A\to A'$ be their join in $\text{CoSub}_A$.
  We assume without loss of generality that $I$ is non-empty, as
  we already know that $\mathbb{Q}_\mathcal{A}$ has a bottom element, namely $\text{id}_A$.
  As $T$ preserves joins, $Te\colon TA\to TA'$ is the join of
  $T(e_i)\colon TA\to T(A_i)$ in $\text{CoSub}_{TA}$. As the square 
  \begin{center}
    \begin{tikzpicture}[modal]
      \node[] at (0,0) (a1) [label=above:{}] {$a^\ast A_{i}$};
      \node[] at (2,0) (b1) [label=above:{}] {$A_{i}$};
      \node[] at (0,2) (e1) [label=above:{}] {$TA$};
      \node[] at (2,2) (f1) [label=above:{}] {$TA_{i}$};
      \path[mono] (a1) edge node[below]{$m_i$} (b1);
      \path[epi] (e1) edge node[left]{$a^\ast e_i$} (a1);
      \path[epi] (e1) edge node[]{$T(e_i)$} (f1);
      \path[->] (f1) edge node[]{$a_i$} (b1);
      \path[->,dashed] (f1) edge node[]{$c_i$} (a1);
    \end{tikzpicture}
  \end{center}
  commutes, the filler property implies that
  $c_i\colon T(e_i)\leq a^\ast(e_i)$ for all $i$.  The map $a^\ast$ is
  order-preserving, so $T(e_i)\leq a^\ast(e_i)\leq a^\ast e$ and as
  $Te$ is the join of the $T(e_i)$, we have $d\colon Te\leq a^\ast e$.
  By postcomposing $d$ with $m\colon a^\ast A'\to A'$ we obtain a
  morphism $a'\colon TA'\to A'$. By construction, $e$ and all
  $e_i\leq e$ are $T$-morphisms if $a'$ is a $T$-algebra. It remains to show
  that $a'$ interacts well with the
  unit and multiplication of the monad.  For the unit,
  $ a'\eta_{A'}e = a'T(e)\eta_A = e a \eta_A = e$ and
  $a'\eta_{A'}=\text{id}_{A'}$, as $e$ is epi. For the
  multiplication, we find that
  \[
    a'\mu_{A'}TT(e) = a'T(e)\mu_A=ea\mu_A=eaT(a)=a'T(e)T(a)=a'T(a')TT(e),
  \]
  and $a'\mu_{A'}=a' T(a')$ as $TT(e)$ is epi.
  So $e\colon \mathcal{A}\to \mathcal{A}'$ is the join of the $e_i$ and $\overline{U}_\mathcal{A}$
  creates all joins in $\text{CoSub}_A$ that are preserved by
  $T$.
\end{proof}

The existence of a right adjoint follows immediately from the general adjoint
functor theorem \cite{maclane98}.

\begin{corollary}\label{cor:uHasRightAdj}
  Let $\mathbb{C}$ be co-well-powered, let $\text{CoSub}_A$ have
  arbitrary joins and let $T$ preserve these joins. Then $\overline{U}_\mathcal{A}$ has a
  right-adjoint.
\end{corollary}

\begin{proof}
  Proposition $\ref{prop:quotJoins}$ shows that $\mathbb{Q}_\mathcal{A}$ has arbitrary joins and these
  are clearly preserved by $\overline{U}_\mathcal{A}$.
  Moreover, as $\mathbb{C}$ is co-well-powered, $\mathbb{Q}_\mathcal{A}$ is
  locally small, and the
  solution set condition holds. Hence  the general adjoint functor
  theorem implies that $\overline{U}_\mathcal{A}$ has a right adjoint.
\end{proof}

Examples \ref{eg:finitaryMonadsLanguage}, \ref{eg:strongFinitaryMonadsLanguage}
and \ref{eg:semiautomataLanguage} satisfy the assumptions of this
corollary. Hence every language in these examples admits a syntactic
quotient.

\subsection{\slominski{}'s Construction}

While the results that follow are less general than the existence
result in Theorem~\ref{thm:weakSyntacticObj}, the technique we refer to as
\emph{\slominski'{}s construction} allows a constructive description
of the right adjoint needed for the existence of syntactic quotients.
In particular, these results can be stated in the language of
fibrations.

In the context of Proposition \ref{prop:quotCongEq}, the forgetful functor
$\overline{U}_\mathcal{A}\colon \mathbb{Q}_\mathcal{A}\to \text{CoSub}_A$
can be recast as a lattice homomorphism
$\overline{U}_\mathcal{A}\colon \text{Cong}(\mathcal{A})\to \text{EqRel}(A)$,
where $\text{EqRel}(A)$
is the complete lattice of equivalence relations on $A$.
For lattices, a right adjoint
$\overline{R}_\mathcal{A}\colon \text{EqRel}(A)\to \text{Cong}(\mathcal{A})$
corresponds to an interior operator on the lattice of equivalence relations.
\slominski{} characterises this interior operator
using \emph{translations} \cite{slominski:1974:greatestCongruence} (for
more details on translations see \cite{bergman:2011:universalAlgebra}).

Let $T=(\Sigma,E)$ be an equational theory with signature $\Sigma$ and
$\mathcal{A}$ a $T$-algebra. The set of \emph{elementary translations}
of $\mathcal{A}$ is given by
\[
  \text{ETr}(\mathcal{A}) = \{\lambda x.f^\mathcal{A}(\ldots,a_{i-1},x,a_{i+1},\ldots)\mid (f,n)\in \Sigma, 0 \leq i < n, a_j\in A\}.
\] The set $\text{Tr}(\mathcal{A})$ of \emph{translations} is the
carrier of the submonoid of $(A^A,\text{id}_A,\circ)$ generated by
$\text{ETr}(\mathcal{A})$.  One may think of $\text{Tr}(\mathcal{A})$
as all the contexts over $\Sigma$ and $A$ that have one hole.

The next lemma connects translations with congruences.

\begin{lemma}\label{lem:translationsAndCong}
  Let $\mathcal{A}$ be a $T$-algebra for a finitary monad $T$ on $\mathbf{Set}$
  and $\sim$ an equivalence  on
  $A$. Then the following are equivalent:
  \begin{enumerate}
    \item $\sim$ is a congruence.
    \item $\forall a,b\in A. \left(a\sim b \implies \forall f\in \text{ETr}(\mathcal{A}).~f(a)\sim f(b)\right)$.
    \item $\forall a,b\in A. \left(a\sim b \implies \forall f\in \text{Tr}(\mathcal{A}).~f(a)\sim f(b)\right)$.
    \end{enumerate}
\end{lemma}

\begin{proof}
  The implications from $(1)$ to $(3)$
  and from $(3)$ to $(2)$ are obvious. It thus remains to show that $(2)$
  implies $(1)$.  Suppose $(1)$ holds. Let $(f,n)\in \Sigma$ and
  $a_i\sim b_i$. We must show that
  $f(a_0,\ldots,a_{n-1})\sim f(b_0,\ldots,b_{n-1})$. It follows from
  that
  $f(\ldots,a_{i-1},a_i,b_{i+1}\ldots)\sim
  f(\ldots,a_{i-1},b_i,b_{i+1},\ldots)$ for each $i$.  The
  conclusion then follows from transitivity of $\sim$.
\end{proof}

The following proposition translates \slominski{}'s construction
\cite[Theorem 1]{slominski:1974:greatestCongruence} to the
language of fibrations.

\begin{proposition}\label{prop:slominski}
  Let $T$ be a finitary monad on $\mathbf{Set}$ and $\mathcal{A}$ a $T$-algebra.
  The forgetful functor
  $\overline{U}_\mathcal{A}\colon \mathbb{Q}_\mathcal{A}\to \text{CoSub}_A$
  has a right-adjoint $(-)^\circ\colon \text{CoSub}_A\to \mathbb{Q}_\mathcal{A}$
  given by
  \[
    e^\circ = \bigwedge_{f\in \text{Tr}(\mathcal{A})} f^\ast e.
  \]
  The reindexing is performed in $\text{CoSub}_A$.
\end{proposition}

\begin{proof}
  By Proposition \ref{prop:quotCongEq} the fibre $\mathbb{Q}_\mathcal{A}$ 
  is equivalent to the lattice of congruences $\text{Cong}(\mathcal{A})$ 
  and $\mathbb{Q}_A$ is equivalent to the lattice of equivalence relations
  $\text{EqRel}(A)$. So instead of working with quotients and left morphisms
  we work with congruences and equivalence relations.

  Let $\sim$ be an equivalence relation on $A$.
  First we show
  that $\sim^\circ$ is again an equivalence.  For all $a\in A$ and
  $f\in \text{Tr}(\mathcal{A})$ it is clear that $f(a)\sim f(a)$ and
  therefore $a\sim^\circ a$. Symmetry and transitivity of $\sim^\circ$
  follow immediately from that of $\sim$.  Moreover, $\sim^\circ$ is a
  congruence by Lemma \ref{lem:translationsAndCong}.

  Finally, we show that $\sim^\circ$ is the greatest congruence
  contained in $\sim$. Obviously,
  $\text{id}_A\in \text{Tr}(\mathcal{A})$, so that
  $a\sim^\circ b \implies a\sim b$ and $\sim^\circ$ is contained in
  $\sim$.  Now assume that ${\equiv}\subseteq {\sim}$ is a
  congruence. If $a\equiv b$, then $f(a) \equiv f(b)$ for all
  $f\in \text{Tr}(\mathcal{A})$ by Lemma
  \ref{lem:translationsAndCong}. As ${\equiv}\subseteq {\sim}$ for
  each translation $f$, we also have $f(a)\sim f(b)$ and hence
  $a\sim^\circ b$ by
  definition. So ${\equiv}\subseteq {\sim^\circ}$.
\end{proof}

The constructive aspect
of this proposition is illustrated by the following examples.
The syntactic congruences considered are classical
\cite{wilke:1993:algebraic,almeida:1990:pseudovarieties,urbat:2017:eilenbergTheorems}; here
we compute them using \slominski{}'s interior operator.

\begin{example}\label{eg:finitaryMonadsSyntactic}
  Let $T$ be the list monad, so that $T$-algebras are monoids.
  An elementary translation of the free monoid $A^\ast$ on $A$ has the form
  $\lambda x.\ xv$ or $\lambda x.\ vx$ with $v\in A^\ast$.
  The translations of $A^\ast$ are therefore of the form
  $\lambda x.\ uxv$ with $u,v\in A^\ast$. An element of
  $\text{CoSub}_{A^\ast}$ is an equivalence $\sim$ on
  $A^\ast$. Applying the definition of $(-)^\circ$, \slominski{}'s
  construction yields
  \[
    {\sim^\circ} = \left(\bigcap_{f\in \text{Tr}(A^\ast)} f^{-1}(\sim) \right)
    = \left\{(x,y)\in A^2\mid \forall u,v\in A^\ast.\ uxv\sim uyv\right\}.
  \]
  Hence $x \sim^\circ y \iff \forall u,v\in A^\ast.\ uxv \sim uyv$.
  Given a language $L\subseteq A^\ast$, we obtain its syntactic quotient
  in $\text{EqRel}(A^\ast)$ by taking the epi-mono factorisation
  $m\circ e$ of $\chi_L\colon A^\ast\to 2$ and then the kernel
  of $e$, where $\ker e = \{(x,y)\mid x\in L \iff y\in L\}$.
  Theorem \ref{thm:weakSyntacticObj} tells us that applying $(-)^\circ$ to
  $\ker e$ yields the syntactic (monoid) congruence $\equiv_L$ of $L$:
  \[
    x \equiv_L y \iff \forall u,v\in A^\ast.\ (uxv \in L \iff uyv\in L).
  \]
  This coincides with the syntactic congruence of $L$ from algebraic
  language theory.
\end{example}

\begin{example}\label{eg:semiautomataSyntactic}
  Instantiating Proposition \ref{prop:slominski} and Theorem
  \ref{thm:weakSyntacticObj} to the semiautomaton
  $\mathcal{S}=(A^\ast,\epsilon,\sigma(u,a)=ua)$ allows deriving the
  Myhill-Nerode congruence of a language. Recall that $\mathcal{S}$ is
  the free $G^\ast$-algebra over $\emptyset$, where $G^\ast$ is the
  free monad generated by $G$, and that $G^\ast$ is associated with
  the equational theory whose signature contains one constant $c$ and
  one unary function symbol $a$ for each $a\in A$.
  Hence the elementary translations of $\mathcal{S}$
  are of the form $\lambda x.\ xu$ with $u\in A^+$, and the translations
  are of the form $\lambda x.\ xu$ with $u\in A^\ast$. For $L\subseteq A^\ast$,
  the syntactic congruence of $L$ in $\text{RCong}(A^\ast)$ is the right congruence
  \[
    x \approx_L y \iff \forall u\in A^\ast. \left( xu\in L \iff yu\in L \right).
  \] So $\approx_L$ is the Myhill-Nerode congruence of $L$, as expected.
\end{example}

The proof of Proposition \ref{prop:slominski} can easily be adapted to
the ordered and multisorted settings. We state the next proposition
as a paradigmatic case and look briefly at two further examples,
namely ordered monoids and Wilke algebras.

\begin{proposition}\label{prop:slominskiPosMulti}
  Proposition \ref{prop:slominski} also holds for strongly finitary monads on $\mathbf{Pos}$
  and for finitary monads on $\mathbf{Set}^k$ (for $k\in \mathbb{N}$).
\end{proposition} 

\begin{proof}
  The finitary monads on $\mathbf{Set}^k$ correspond to multi-sorted
  equational theories, while the strongly finitary monads on $\mathbf{Pos}$
  correspond to inequational theories with classical signatures.
  In both cases, the definition of translation can easily be
  adapted. 

  For the multisorted setting, the proof of \slominski{}'s result generalises
  immediately. For the ordered setting, we are working with precongruences
  instead of congruences. Here we note that the proof never makes use of
  symmetry, and hence generalises again.
\end{proof}

We instantiate Proposition \ref{prop:slominskiPosMulti} in the ordered
setting to show how the syntactic precongruence for a language on the
free ordered monoid $(A^\ast,=)$ can be derived. We omit the details
for this example; the development follows closely that of Examples
\ref{eg:finitaryMonadsSyntactic} and
\ref{eg:semiautomataSyntactic}. 

\begin{example}
  Let $(A^\ast,=)$ be the free ordered monoid on $(A,=)$. The translations
  are precisely those of the free monoid $A^\ast$, that is, the functions
  of the form $\lambda x.\ uxv$ for $u,v\in A^\ast$. The syntactic
  precongruence of a language $L\subseteq A^\ast$ is concretely given by
  \[
    x \lesssim_L y \iff \forall u,v\in A^\ast.\left( uxv\in L \implies uyv\in L \right).
  \]
  In a similar way, we can derive a description of the ordered Myhill-Nerode
  congruence \cite{klima:2019:varietiesOfOrderedAutomata}.
\end{example}

Finally, we  instantiate Proposition \ref{prop:slominskiPosMulti} in
the multisorted setting to derive the syntactic Wilke algebra congruence
for an $\infty$-language $L\subseteq A^{\infty}$.

\begin{example}
  Consider the set $A^\omega$ of all infinite words over the set
  $A$. We write $A^\text{up}$ for the ultimately periodic
  words, which are of the form $uv^\omega$.  The two-sorted
  signature of Wilke algebras over sorts
  $S=\{\mathsf{f},\mathsf{i}\} $ consists of the three function
  symbols
  $(- \cdot -)\colon \mathsf{f} \times \mathsf{f} \to \mathsf{f}$,
  $(- \times -)\colon \mathsf{f} \times \mathsf{i} \to \mathsf{i}$ and
  $(-)^\omega\colon \mathsf{f}\to \mathsf{i}$.  The \emph{free Wilke algebra}
  over $A$ has carrier $(A^+,A^{\text{up}})$, with
  operations $u\cdot v=uv$, $u\times vw^\omega = uvw^\omega$ and
  $u^\omega = u^\omega$. We are interested in the syntactic congruence
  for $\infty$-languages, however for technical reasons we require that 
  the subset of infinite words is uniquely determined by the
  ultimately periodic words it contains. The generalisation of
  translations to the two-sorted setting is straightforward, one
  forms them as before but ensures that sorts match.

  For the first sort, there are three possible applications. A word $u$
  can either be multiplied by a finite word (on the left or right),
  by an infinite word (on the right), or be repeated infinitely often by
  applying $(-)^\omega$. For the second sort, the only possible operation 
  is left multiplication by a finite word. Given a language
  $L_0\uplus L_1 = L\subseteq A^\infty$
  with $L_0\subseteq A^+$ and $L_1 \subseteq A^{\omega}$, and applying Proposition \ref{prop:slominskiPosMulti}
  to the kernel of its characteristic function, we obtain the Wilke
  algebra congruence \cite{wilke:1993:algebraic}
  \begin{align*}
    u \equiv_1 v &\iff \forall x,y\in A^\ast.\left( xuy\in L_0 \iff xvy\in L_0 \right) \text{ and} \\
                 &\qquad\quad\forall x\in A^\ast,\alpha\in A^{\text{up}}.\left(xu\alpha\in L_1 \iff xv\alpha\in L_1\right) \text{ and }\\
    &\qquad\quad \forall x,y,z\in A^\ast.\left(x(yuz)^\omega\in L_1 \iff x(yvz)^\omega \in L_1\right), \\
    \alpha \equiv_2\beta &\iff \forall x\in A^\ast.\left(x\alpha\in L_1 \iff x\beta\in L_1\right).
  \end{align*}
  Using the laws for $(-)^\omega$, we can simplify the third line
  to
  \[
    \forall x,y\in A^\ast.\left( x(uy)^\omega\in L \iff x(vy)^\omega\in L \right).
  \]
  The congruence described here is the syntactic congruence of $L$ for
  Wilke algebras
  \cite{wilke:1993:algebraic}.
\end{example}

\section{Properties of Recognition and (Regular) Languages}

In language theory, regular languages are recognised by quotients with
finite codomain. Here we assume that the underlying category
$\mathbb{C}$ contains a distinguished class of `finite' objects, which
allows us to define a generic notion of regularity. We then propose a
sufficient condition under which regular languages are stable under
reindexing and study properties of the recognition fibration. The main
result of this section provides conditions under which regular
languages over $T$-algebras are closed under $\mathcal{J}$-limits and
$\mathcal{J}$-colimits. In concrete examples, this leads to
Boolean closure properties of regular languages. We state all results
relative to a monad $T$, a $T$-algebra $\mathcal{A}$, a finite
quotient $e\colon \mathcal{A}\leftmorph \mathcal{B}$ and a language
$p\colon A\to \Omega$.

\begin{definition}
  A language $p\colon A\to \Omega$ in $\mathbb{L}_\mathcal{A}$ is
 \emph{regular} if and only if it is recognised by some quotient
  $e\colon \mathcal{A}\leftmorph \mathcal{B}$ in
  $\mathbb{Q}_\mathcal{A}$ with finite codomain.
\end{definition}

Let $\mathbb{Q}_\text{fi}(T)$ be the full subcategory of $\mathbb{Q}$
on pairs $(\mathcal{A},e)$ for which the codomain of $e$ is
finite. In the presence of further assumptions, regularity is
preserved under reindexing. We say that right morphisms
\emph{reflect finiteness} if whenever $B$ is finite and
$m\colon A\rightmorph B$, $A$ is also finite.

\begin{proposition}
  Let right morphisms reflect finiteness. Then
  $\dom_{\mathbb{Q}_{\text{fi}}}\colon\mathbb{Q}_{\text{fi}}(\mathbb{T})\to \mathbb{C}^T$
  is a subfibration of the quotient fibration.
\end{proposition}

\begin{proof}
  Let $e\colon \mathcal{A}\leftmorph \mathcal{B}$ with $B$ finite.  We
  show that the Cartesian morphism
  $\overline{f}(e)\colon f^\ast e\to e$ exists in
  $\mathbb{Q}_{\text{fi}}$.  Recall that $\overline{f}(e)=(f,m)$,
  where $m$ is the right morphism obtained by factoring
  $e\circ f=m\circ f^\ast e$. The right-hand side of this equation is
  a finite quotient whenever the domain of $m$ is finite.  As $B$ is
  finite and $m$ reflects finiteness, this is the case.
\end{proof}

For a uniform fibrational treatment of closure properties of regular
languages we assume from here on that
$\dom_\mathbb{F}\colon \mathbb{F}(T)\to \mathbb{C}^T$ is a
subfibration of the quotient fibration. We think of the fibre
$\mathbb{F}_\mathcal{A}$ as consisting of finite quotients over
$\mathcal{A}$. Moreover, for the remainder of this section we write
$\mathbb{R}$ for the category whose objects are pairs
$(\mathcal{A},e\colon \mathcal{A}\leftmorph \mathcal{B},p\colon A\to
\Omega)$ with $e\in \mathbb{F}_\mathcal{A}$ recognising
$p\in \mathbb{L}_\mathcal{A}$.  Our use of $\mathbb{Q}$ and
$\mathbb{L}$ remains unchanged. Finally, we write
$\mathbb{L}_{\text{reg}}$ for the full subcategory of $\mathbb{L}$
consisting of pairs $(\mathcal{A},p\colon A\to\Omega)$ with $p$
recognised by some $e\in \mathbb{F}_\mathcal{A}$. Hence for a
$T$-algebra $\mathcal{A}$, the fibre
$(\mathbb{L}_{\text{reg}})_\mathcal{A}$ consists of all regular
languages over $\mathcal{A}$. By construction, the restriction
$L_\text{reg}\colon \mathbb{L}_\text{reg}\to \mathbb{C}^T$ of the language
fibration $L\colon \mathbb{L}\to \mathbb{C}^T$ to $\mathbb{L}_\text{reg}$
is a fibration.

The first of the following facts shows that the finite quotients
recognising a language $p$ are down-closed in $\mathbb{F}$. This means
that if $e$ recognises $p$ and $c\colon e'\to e$, then $e'$ also
recognises $p$.  The mere existence of $c$ therefore suffices to
establish that $e'$ also recognises $p$.

\begin{proposition}\label{prop:recMDownClosed}
  The fibre $\mathbb{R}(T)_{(\mathcal{A},p)}$ is down-closed in $\mathbb{F}(T)$.
  If $\text{id}_A\colon \mathcal{A}\to \mathcal{A}\in \mathbb{F}_\mathcal{A}$,
  then $\mathbb{R}_{(\mathcal{A},p)}$ has an initial object.
\end{proposition} 

\begin{proof}
  First, let $e\colon \mathcal{A}\leftmorph \mathcal{B}\in \mathbb{R}_{(\mathcal{A},p)}$ and
  $w\colon e'\to e$ for $e'\colon \mathcal{A}\leftmorph \mathcal{C}\in \mathbb{F}_\mathcal{A}$,
  so $w\colon \mathcal{C}\leftmorph \mathcal{B}$ in $\mathbb{C}^T$.
  We wish to show that $e'\in \mathbb{R}_{(\mathcal{A},p)}$.
  As $e$ recognises $p$, there is some $q\in \mathbb{L}_\mathcal{B}$
  such that $p=q\circ e$. For $q\circ w\in \mathbb{L}_\mathbb{C}$ we
  have that $q\circ w \circ e' = q\circ e = p$ and hence
  $e'\in \mathbb{R}_{(\mathcal{A},p)}$ as required.
  
  Second, initiality of $\text{id}_A$ is straightforward.
\end{proof}

As a corollary, we obtain the following existence result for limits and
colimits in $\mathbb{R}_{(\mathcal{A},p)}$.

\begin{corollary}\label{cor:forgetfulCreatesCoLimits}
  The forgetful functor $(U_\mathbb{F})_{(\mathcal{A},p)}\colon \mathbb{R}(T)_{(\mathcal{A},p)}\to \mathbb{F}(T)_\mathcal{A}$
  creates all non-empty limits.
  If $\mathbb{R}(T)_{(\mathcal{A},p)}$ has a weakly terminal object,
  $(U_\mathbb{F})_{(\mathcal{A},p)}$ also creates
  all non-empty colimits.
\end{corollary}

\begin{proof}
  First we show that the forgetful functor creates all non-empty limits.
  Let $D\colon I\to \mathbb{R}_{(\mathcal{A},p)}$ be a non-empty diagram and
  let $(e\colon \mathcal{A}\to \mathcal{B},\sigma)$ be a limiting cone
  for $(U_\mathbb{F})_p\circ D$. Take any $e_i\colon \mathcal{A}\to \mathcal{A}_{i}$.
  As the diagram is non-empty, there is a morphism
  $\sigma_i\colon e\to e_i$. As $e_i\in \mathbb{R}_{(\mathcal{A},p)}$
  and $\mathbb{R}_{(\mathcal{A},p)}$ is
  down-closed by
  Proposition \ref{prop:recMDownClosed}, $e\in
  \mathbb{R}_{(\mathcal{A},p)}$ as well and
  $(e,\sigma)$ is limiting in $\mathbb{R}_{(\mathcal{A},p)}$.

  Next suppose $e_t$ is the terminal object in $\mathbb{R}_{(\mathcal{A},p)}$.
  We show that $(U_\mathbb{F})_p$ creates all non-empty colimits.
  Let $D\colon I\to \mathbb{R}_{(\mathcal{A},p)}$ be a non-empty diagram,
  $(e\colon \mathcal{A}\to \mathcal{B},\sigma)$ be the colimiting cocone
  for $(U_\mathbb{F})_p\circ D$.
  For each $e_i$ in the diagram $D$, let $t_i\colon e_i\to e_t$ be the
  unique map to the terminal object in $\mathbb{R}_{(\mathcal{A},p)}$. Then the
  $t_i\colon e_i\to e_t$ are also maps in $\mathbb{F}_\mathcal{A}$ and
  by the universal property of $e$, there is a morphism
  $c\colon e\to e_t$ in $\mathbb{F}_\mathcal{A}$. Finally, as $e_t\in \mathbb{R}_{(\mathcal{A},p)}$
  and $\mathbb{R}_{(\mathcal{A},p)}$ is down-closed by Proposition
  \ref{prop:recMDownClosed}, $e\in \mathbb{R}_{(\mathcal{A},p)}$ is colimiting.
\end{proof}

The next proposition generalises well-known connections between
regularity and  finiteness of syntactic quotients.

\begin{proposition}\label{prop:regularSyntactic}
  Let $\overline{U}_\mathcal{A}\colon \mathbb{Q}(T)_\mathcal{A}\to \text{CoSub}_A$ have
  a right adjoint $\overline{R}_\mathcal{A}$,
  $\mathbb{R}(\text{Id}_\mathbb{C})_{(\mathcal{A},p)}$ have a terminal object
  $e_t$ and $\mathbb{F}_\mathcal{A}$ be upward-closed in $\mathbb{Q}_\mathcal{A}$.
  Then $p$ is regular if and only if $\overline{R}_{\mathcal{A}}(e_t)\in \mathbb{F}(T)_\mathcal{A}$.
\end{proposition}

\begin{proof}
  If $p$ is regular, then there exists some $e\in \mathbb{R}_{(\mathcal{A},p)}$.
  As there exists a unique arrow $e \to \overline{R}_\mathcal{A}(e_t)$
  ($\overline{R}_\mathcal{A}$ preserves terminal objects)
  in $\mathbb{Q}_\mathcal{A}$ and $\mathbb{F}_\mathcal{A}$ is upward closed,
  $\overline{R}_\mathcal{A}(e_t)\in \mathbb{F}_\mathcal{A}$.

  For the converse direction, we know that $\overline{R}_{\mathcal{A}}(e_t)$
  recognises $p$. Thus $p$ is regular by definition.
\end{proof}

Finally, we present sufficient conditions under which regular
languages are closed under $\mathcal{J}$-limits or
$\mathcal{J}$-colimits. The ideas behind the proof can be understood
through semiautomata. It is well-known that the product construction
on semiautomata can be used to accept both the union and intersection
of two regular languages, considering unions and intersections of
accepting states, respectively. A similar construction applies to
$\mathcal{J}$-limits and $\mathcal{J}$-colimits. 

\begin{theorem}\label{thm:limitsColimitsRegular}
  Let the fibre $\mathbb{F}_\mathcal{A}$ have a cone for every discrete diagram of size $|\mathcal{J}|$.
  If the language fibration $L$ has fibred $\mathcal{J}$-limits
  ($\mathcal{J}$-colimits), then the fibre $(\mathbb{L}_{\text{reg}})_\mathcal{A}$
  has $\mathcal{J}$-limits ($\mathcal{J}$-colimits).
\end{theorem}

\begin{proof}
  We only show the theorem for $\mathcal{J}$-limits, the proof for
  $\mathcal{J}$-colimits is analogous.
  Let $\mathcal{J}$ be a diagram in
  $(\mathbb{L}_{\text{reg}})_\mathcal{A}$ which has a limit
  $p$ in $\mathbb{L}_\mathcal{A}$. For each language
  $p_i$ in the diagram, assume that $e_i\in \mathbb{F}_\mathcal{A}$
  recognises it.  As there are $\mathcal{J}$-many $e_i$, there exists
  a cone $(e\colon \mathcal{A}\leftmorph \mathcal{B},\sigma)$ over the
  $e_i$. Now $e_i\in \mathbb{R}_{(\mathcal{A},p_i)}$ and
  $\sigma_i\colon e\to e_i$ for all $i$, so that
  $e\in \mathbb{R}_{(\mathcal{A},p_i)}$ for all $i$ as
  $\mathbb{R}_{(\mathcal{A},p_i)}$ is down-closed by Proposition
  \ref{prop:recMDownClosed}. Equivalently,
  therefore, $p_i\in \mathbb{R}_{(\mathcal{A},e)}$ for all $i$.
  Let $q_i\in \mathbb{L}_\mathcal{B}$ be such that
  $p_i = q_i\circ e = e^\ast q_i$.  As
  $\dom_{\mathbb{L}}$ has fibred $\mathcal{J}$-limits, the $\mathcal{J}$-limit
  $q$ exists and moreover it is
  preserved by $e^\ast$. Hence
  \[
    e^\ast q = e^\ast(\lim_\mathcal{J} q_i)=\lim_\mathcal{J} (e^\ast q_i) = \lim_\mathcal{J} p_i=p.
  \]
  It follows that $p$ is recognised by $e\in \mathbb{F}_\mathcal{A}$
  and hence $p\in (\mathbb{L}_{\text{reg}})_\mathcal{A}$.
\end{proof}

This theorem only considers regular languages over a specific $T$-algebra.
The following corollary establishes it more globally for fibrations.

\begin{corollary}
  Let the finite quotient fibration $\dom_\mathbb{F}$ have a
  cone for every discrete diagram of size
  $|\mathcal{J}|$ in each fibre. If the language fibration $L$ has fibred
  $\mathcal{J}$-limits ($\mathcal{J}$-colimits), then so does the regular language
  fibration $L_\text{reg}$.
\end{corollary}

\begin{proof}
  The reindexing of $L_\text{reg}$ is the same as that of $L$, the
  corollary then follows immediately from Theorem
  \ref{thm:limitsColimitsRegular}.
\end{proof}

\begin{example}
  In $\mathbf{Set}$ we choose finite sets as finite objects.
  These interact as follows with the factorisation system:
  if $e\colon A\twoheadrightarrow B$
  and $A$ is finite so is $B$, and if $m\colon A\hookrightarrow B$ and
  $B$ is finite, so is $A$. A surjective homomorphism
  $e\colon \mathcal{A}\twoheadrightarrow \mathcal{B}$ with finite codomain corresponds
  to a congruence on $\mathcal{A}$ that has finitely many equivalence classes
  (a congruence of \emph{finite index}). Hence the fibre
  $\mathbb{F}_\mathcal{A}$ is equivalent to the lattice of congruences of
  $\mathcal{A}$ of finite index (with inclusions
  as morphisms). This lattice has finite meets
  and arbitrary non-empty joins.

  Let $\mathcal{A}$ be a $T$-algebra and $L\subseteq \mathcal{A}$. As
  each language has a syntactic quotient, the sublattice of
  congruences of finite index that saturate $L$ has finite meets and
  arbitrary non-empty joins as well by Corollary
  \ref{cor:forgetfulCreatesCoLimits}. Moreover, $L$ is regular if and
  only if its syntactic congruence is of finite index by Proposition
  \ref{prop:regularSyntactic}.  As the fibre $\mathbb{F}_\mathcal{A}$
  only has finite meets, regular languages are closed under finite
  intersections and finite unions by Theorem
  \ref{thm:limitsColimitsRegular}.

  Each language fibre $\mathbb{L}_\mathcal{B}$ forms a
  Boolean algebra and is complemented. Moreover, complementation
  is preserved by reindexing. One can then show that regular languages are
  also closed under complementation.
\end{example}

\begin{example}
  In the category $\mathbf{Pos}$, consider finite posets as finite objects.
  Again the factorisation system interacts
  well with this  notion of finiteness. As in the previous example,
  we find that the fibre $\mathbb{F}_\mathcal{A}$ consists of precongruences
  whose induced equivalence relation is of finite index. The same results
  apply in this context, the only difference being that $\mathbb{L}_\mathcal{B}$
  is only a lattice and not a Boolean algebra. Hence here, regular languages
  are no longer closed under complementation. 
\end{example}

The fibrational setting and our notion of regularity therefore allow us to
formulate closure properties of regular languages uniformly in
different settings.

\section{Conclusion}

We gave a fibrational account of language recognition by quotients
based on the idea that recognition is stable under reindexing.  For a
monad $T$ on a category $\mathbb{C}$ equipped with a factorisation system,
we introduced formal definitions for quotients and languages.
We identified sufficient conditions for
the existence of the corresponding quotient and language fibrations,
thus ensuring that our abstract notion of recognition is stable under
reindexing.  This in turn gives rise to a combined recognition
fibration (see Figure~\ref{eq:story}).  Moreover, we proved an
existence theorem for syntactic quotients for languages in the
presence of a factorisation system and analysed how it relates to
\slominski{}'s construction of a greatest congruence contained in an
equivalence.  We then expressed \slominski{}'s construction
in the language of fibrations and used it, as in our existence
theorem, for deriving syntactic congruences for paradigmatic instances,
including ordered monoids and Wilke algebras. Finally, we defined
regular languages relative to a distinguished class of finite objects
and provided necessary conditions under which regular languages are
closed under limits and colimits.

We envisage the following directions for further work.  First, we
currently discuss \slominski{}'s construction only for varieties of
(ordered/multisorted) varieties. Instead a
fully fibrational treatment in the presence of a factorisation system
would be desirable. Boja\'nczyk's use of polynomials in
\cite{bojanczyk:2015:recognisableLanguagesOverMonads} provides a
promising starting point. Second, beyond the fibrational derivation of
Boolean closure properties of languages, a study of algebraic closure
properties remains to be undertaken. Third, our abstract definition of
languages depends on the parameter $\Omega$ that leaves many
choices. Particularly interesting seem quantitative rather than
qualitative notions of language as for instance the formal power
series (see \cite{droste:2009:semiringsAndFormalPowerSerier}) used in
combination with weighted or probabilistic automata.

\bibliographystyle{plain}
\bibliography{references}

@incollection {droste:2009:semiringsAndFormalPowerSerier,
    AUTHOR = {Droste, Manfred and Kuich, Werner},
     TITLE = {Chapter 1: {S}emirings and formal power series},
 BOOKTITLE = {Handbook of weighted automata},
    SERIES = {Monogr. Theoret. Comput. Sci. EATCS Ser.},
     PAGES = {3--28},
 PUBLISHER = {Springer, Berlin},
      YEAR = {2009},
      ISBN = {978-3-642-01491-8},
   MRCLASS = {68Q70},
  MRNUMBER = {2777727},
       DOI = {10.1007/978-3-642-01492-5\_1},
       URL = {https://doi.org/10.1007/978-3-642-01492-5_1},
}

@book{johnson:2021:2DimensionalCategories,
    AUTHOR = {Johnson, Niles and Yau, Donald},
     TITLE = {2-dimensional categories},
 PUBLISHER = {Oxford University Press, Oxford},
      YEAR = {2021},
     PAGES = {xix+615},
      ISBN = {978-0-19-887138-5; 978-0-19-887137-8},
   MRCLASS = {18-02 (18N10 18N15 18N20)},
  MRNUMBER = {4261588},
MRREVIEWER = {Robert\ Laugwitz},
       DOI = {10.1093/oso/9780198871378.001.0001},
       URL = {https://doi.org/10.1093/oso/9780198871378.001.0001},
}

@incollection {urbat:2017:eilenbergTheorems,
    AUTHOR = {Urbat, Henning and Ad\'amek, Ji\v{r}\'i{} and Chen, Liang-Ting
              and Milius, Stefan},
     TITLE = {Eilenberg theorems for free},
 BOOKTITLE = {42nd {I}nternational {S}ymposium on {M}athematical
              {F}oundations of {C}omputer {S}cience},
    SERIES = {LIPIcs. Leibniz Int. Proc. Inform.},
    VOLUME = {83},
     PAGES = {Art. No. 43, 15},
 PUBLISHER = {Schloss Dagstuhl. Leibniz-Zent. Inform., Wadern},
      YEAR = {2017},
      ISBN = {978-3-95977-046-0},
   MRCLASS = {68Q70},
  MRNUMBER = {3755336},
}

@article {almeida:1990:pseudovarieties,
    AUTHOR = {Almeida, Jorge},
     TITLE = {On pseudovarieties, varieties of languages, filters of
              congruences, pseudoidentities and related topics},
   JOURNAL = {Algebra Universalis},
  FJOURNAL = {Algebra Universalis},
    VOLUME = {27},
      YEAR = {1990},
    NUMBER = {3},
     PAGES = {333--350},
      ISSN = {0002-5240,1420-8911},
   MRCLASS = {08C99 (08B05 08B15)},
  MRNUMBER = {1058478},
MRREVIEWER = {James\ B.\ Nation},
       DOI = {10.1007/BF01190713},
}

@article{adamek:1979:cogenerationOfAlgebras,
  title={{O}n the {C}ogeneration of {A}lgebras},
  author={Ji{\v r}{\'i} Ad{\'a}mek},
  journal={Math. Nachr.},
  year={1979},
  volume={88},
  pages={373--384},
  doi = {10.1002/mana.19790880129}
}

@book {adamek:1990:automataAndAlgebras,
    AUTHOR = {Ad{\'a}mek, Ji{\v r}{\'i} and Trnkov{\'a}, V{\v e}ra},
     TITLE = {Automata and algebras in categories},
 PUBLISHER = {Kluwer Academic Publishers Group},
      YEAR = {1990}
}

@book {adamek:1990:automataAndAlgebrasInCategories,
    AUTHOR = {Ad{\'a}mek, Ji{\v r}{\'i} and Trnkov{\'a}, V{\v e}ra},
     TITLE = {Automata and algebras in categories},
 PUBLISHER = {Kluwer Academic Publishers Group},
      YEAR = {1990}
}

@book {adamek:1990:joyOfCats,
    AUTHOR = {Ad\'amek, Ji{\v r}{\'i} and Herrlich, Horst and Strecker, George
              E.},
     TITLE = {Abstract and concrete categories},
 PUBLISHER = {John Wiley \& Sons, Inc.},
      YEAR = {1990}
}

@article {adamek:2018:catApproach,
    AUTHOR = {Ad{\'a}mek, Ji{\v r}{\'i} and Milius, Stefan and Urbat, Henning},
     TITLE = {A categorical approach to syntactic monoids},
   JOURNAL = {Log. Methods Comput. Sci.},
  FJOURNAL = {Logical Methods in Computer Science},
    VOLUME = {14},
      YEAR = {2018},
    NUMBER = {2},
     PAGES = {Paper No. 9, 34},
       DOI = {10.23638/LMCS-14(2:9)2018}
}

@article {adamek:2022:categoricalViewOrderedAlgebras,
    AUTHOR = {Ad\'amek, J. and Dost\'al, M. and Velebil, J.},
     TITLE = {A categorical view of varieties of ordered algebras},
   JOURNAL = {Math. Structures Comput. Sci.},
  FJOURNAL = {Mathematical Structures in Computer Science. A Journal in the
              Applications of Categorical, Algebraic and Geometric Methods
              in Computer Science},
    VOLUME = {32},
      YEAR = {2022},
    NUMBER = {4},
     PAGES = {349--373},
       DOI = {10.1017/S0960129521000463}
}

@article {arbib:1975:adjointMachines,
    AUTHOR = {Arbib, Michael A. and Manes, Ernest G.},
     TITLE = {Adjoint machines, state-behavior machines, and duality},
   JOURNAL = {J. Pure Appl. Algebra},
  FJOURNAL = {Journal of Pure and Applied Algebra},
    VOLUME = {6},
      YEAR = {1975},
    NUMBER = {3},
     PAGES = {313--344},
      ISSN = {0022-4049,1873-1376},
   MRCLASS = {18B20 (94A30)},
  MRNUMBER = {414654},
MRREVIEWER = {S.\ Walig\'orski},
       DOI = {10.1016/0022-4049(75)90028-6},
       URL = {https://doi.org/10.1016/0022-4049(75)90028-6},
}

@article {arbib:1980:machines,
    AUTHOR = {Arbib, Michael A. and Manes, Ernest G.},
     TITLE = {Machines in a category},
   JOURNAL = {J. Pure Appl. Algebra},
  FJOURNAL = {Journal of Pure and Applied Algebra},
    VOLUME = {19},
      YEAR = {1980},
     PAGES = {9--20},
      ISSN = {0022-4049,1873-1376},
   MRCLASS = {68D30 (01A60 18B20 68D15 93Bxx)},
  MRNUMBER = {593243},
       DOI = {10.1016/0022-4049(80)90090-0},
       URL = {https://doi.org/10.1016/0022-4049(80)90090-0},
}

@book {awodey:2010:categoryTheory,
    AUTHOR = {Awodey, Steve},
     TITLE = {Category theory},
   EDITION = {Second},
 PUBLISHER = {Oxford University Press},
      YEAR = {2010},
 DOI = {10.1093/acprof:oso/9780198568612.001.0001}
}

@book {barr:1990:categoryTheory,
    AUTHOR = {Barr, Michael and Wells, Charles},
     TITLE = {Category theory for computing science},
 PUBLISHER = {Prentice Hall International},
      YEAR = {1990}
}

@book{bergman:2011:universalAlgebra,
  title={Universal algebra: Fundamentals and selected topics},
  author={Bergman, Clifford},
  year={2011},
  publisher={Chapman and Hall/CRC},
  doi = {10.1201/9781439851302}
}

@article {bloom:1976:varietiesOfOrderedAlgebras,
    AUTHOR = {Bloom, Stephen L.},
     TITLE = {Varieties of ordered algebras},
   JOURNAL = {J. Comput. System Sci.},
  FJOURNAL = {Journal of Computer and System Sciences},
    VOLUME = {13},
      YEAR = {1976},
    NUMBER = {2},
     PAGES = {200--212},
       DOI = {10.1016/S0022-0000(76)80030-X}
}

@incollection {bojanczyk:2015:recognisableLanguagesOverMonads,
    AUTHOR = {Boja\'nczyk, Miko\l{}aj},
     TITLE = {Recognisable languages over monads},
 BOOKTITLE = {Developments in language theory},
    SERIES = {Lecture Notes in Comput. Sci.},
    VOLUME = {9168},
     PAGES = {1--13},
 PUBLISHER = {Springer},
      YEAR = {2015},
       DOI = {10.1007/978-3-319-21500-6\_1}
}

@book {burris:1981:universalAlgebra,
    AUTHOR = {Burris, Stanley and Sankappanavar, H. P.},
     TITLE = {A course in universal algebra},
 PUBLISHER = {Springer},
      YEAR = {1981},
doi = {10.2307/2322184}
}

@incollection {ciancia:2019:omegaAutomata,
    AUTHOR = {Ciancia, Vincenzo and Venema, Yde},
     TITLE = {{$\Omega$}-automata: a coalgebraic perspective on regular
              {$\omega$}-languages},
 BOOKTITLE = {8th {C}onference on {A}lgebra and {C}oalgebra in {C}omputer
              {S}cience},
    SERIES = {LIPIcs. Leibniz Int. Proc. Inform.},
    VOLUME = {139},
     PAGES = {Art. No. 5, 18},
 PUBLISHER = {Schloss Dagstuhl. Leibniz-Zent. Inform.},
      YEAR = {2019},
  DOI = {10.4230/LIPIcs.CALCO.2019.5}
}

@article {colcombet:2020:automataMinimization,
    AUTHOR = {Colcombet, Thomas and Petri\c{s}an, Daniela},
     TITLE = {Automata minimization: a functorial approach},
   JOURNAL = {Log. Methods Comput. Sci.},
  FJOURNAL = {Logical Methods in Computer Science},
    VOLUME = {16},
      YEAR = {2020},
    NUMBER = {1},
     PAGES = {Paper No. 32, 28},
      ISSN = {1860-5974},
   MRCLASS = {68Q70 (18B20)},
  MRNUMBER = {4086864},
       DOI = {10.23638/LMCS-16(1:32)2020},
       URL = {https://doi.org/10.23638/LMCS-16(1:32)2020},
}

@book{ginzburg:1968:algebraicTheoryOfAutomata,
  title = {{A}lgebraic {T}heory of {A}utomata},
  author = {Ginzburg, Abraham},
  year = {1968},
  publisher = {Academic Press},
  doi = {10.1016/B978-1-4832-0013-2.50012-6}
}

@BOOK{jacobs:1999:categoricalLogic, 
author       = "Jacobs, Bart", 
title        = "{C}ategorical {L}ogic and {T}ype {T}heory",
publisher    = "North Holland", 
address      = "Amsterdam", 
year         = "1999"
}

@book {jacobs:2017:introductionToCoalgebra,
    AUTHOR = {Jacobs, Bart},
     TITLE = {Introduction to coalgebra},
 PUBLISHER = {Cambridge University Press},
      YEAR = {2017},
       DOI = {10.1017/CBO9781316823187}
}

@incollection {klima:2019:varietiesOfOrderedAutomata,
    AUTHOR = {Kl\'ima, Ond{\v r}ej and Pol\'ak, Libor},
     TITLE = {On varieties of ordered automata},
 BOOKTITLE = {Language and automata theory and applications},
    SERIES = {Lecture Notes in Comput. Sci.},
    VOLUME = {11417},
     PAGES = {108--120},
 PUBLISHER = {Springer},
      YEAR = {2019},
       DOI = {10.1007/978-3-030-13435-8\_8}
}

@incollection {linton:1965:aspectsOfEquationalCategories,
    AUTHOR = {Linton, F. E. J.},
     TITLE = {Some aspects of equational categories},
 BOOKTITLE = {Proc. {C}onf. {C}ategorical {A}lgebra},
     PAGES = {84--94},
 PUBLISHER = {Springer},
      YEAR = {1966},
DOI = {10.1007/978-3-642-99902-4_3}
}

@incollection {mellies:2023:parsingAsLifting,
    AUTHOR = {Melli\`es, Paul-Andr\'e{} and Zeilberger, Noam},
     TITLE = {Parsing as a lifting problem and the
              {C}homsky-{S}ch\"utzenberger representation theorem},
 BOOKTITLE = {Mathematical {F}oundations of {P}rogramming
              {S}emantics---{T}hirty-{E}ighth {A}nnual {C}onference},
    SERIES = {Electron. Notes Theor. Inform. Comput. Sci.},
    VOLUME = {1},
     PAGES = {Paper No. 11, 22},
 PUBLISHER = {Episciences},
      YEAR = {2023},
       DOI = {10.46298/entics.proceedings.mfps38}
}

@incollection {mellies:2026:ccfibrationOfHigherOrderRegLang,
    AUTHOR = {Melli\`es, Paul-Andr\'e{} and Moreau, Vincent},
     TITLE = {A {C}artesian {C}losed {F}ibration of {H}igher-{O}rder
              {R}egular {L}anguages},
 BOOKTITLE = {41st {A}nnual {S}ymposium on {L}ogic in {C}omputer {S}cience
              ({LICS} 2026)},
    SERIES = {LIPIcs. Leibniz Int. Proc. Inform.},
    VOLUME = {380},
     PAGES = {Art. No. 73, 25},
 PUBLISHER = {Schloss Dagstuhl. Leibniz-Zent. Inform},
      YEAR = {2026},
       DOI = {10.4230/lipics.lics.2026.73}
}

@article{pin:1995:varietyTheoremWithoutComplementation,
  title={A variety theorem without complementation},
  author={Jean-{\'E}ric Pin},
  journal = {Russ. Math.},
  volume = {39},
  pages = {80--90},
  year={1995}
}

@article {rutten:2000:universalCoalgebra,
    AUTHOR = {Rutten, J. J. M. M.},
     TITLE = {Universal coalgebra: a theory of systems},
   JOURNAL = {Theoret. Comput. Sci.},
  FJOURNAL = {Theoretical Computer Science},
    VOLUME = {249},
      YEAR = {2000},
    NUMBER = {1},
     PAGES = {3--80},
       DOI = {10.1016/S0304-3975(00)00056-6}
}

@article {slominski:1974:greatestCongruence,
    AUTHOR = {S{\l}omi\'nski, J\'ozef},
     TITLE = {On the greatest congruence relation contained in an
              equivalence relation and its applications to the algebraic
              theory of machines},
   JOURNAL = {Colloq. Math.},
  FJOURNAL = {Colloquium Mathematicum},
    VOLUME = {29},
      YEAR = {1974},
     PAGES = {31--43, 159},
      ISSN = {0010-1354,1730-6302},
   MRCLASS = {94A30 (08A05)},
  MRNUMBER = {356978},
MRREVIEWER = {David\ C.\ Rine},
       DOI = {10.4064/cm-29-1-31-43},
       URL = {https://doi.org/10.4064/cm-29-1-31-43},
}

@article{vidal:2020:congruenceBasedProofs,
  author       = {Juan Climent Vidal and
                  Enric Cosme{-}Ll{\'{o}}pez},
  title        = {Congruence-based proofs of the recognizability theorems for free many-sorted
                  algebras},
  journal      = {J. Log. Comput.},
  volume       = {30},
  number       = {2},
  pages        = {561--633},
  year         = {2020},
  doi          = {10.1093/LOGCOM/EXZ032}
}

@article {wilke:1993:algebraic,
    AUTHOR = {Wilke, Thomas},
     TITLE = {An algebraic theory for regular languages of finite and
              infinite words},
   JOURNAL = {Internat. J. Algebra Comput.},
  FJOURNAL = {International Journal of Algebra and Computation},
    VOLUME = {3},
      YEAR = {1993},
    NUMBER = {4},
     PAGES = {447--489},
       DOI = {10.1142/S0218196793000287}
}

@book {maclane98,
    AUTHOR = {Mac Lane, Saunders},
     TITLE = {Categories for the working mathematician},
   EDITION = {Second},
 PUBLISHER = {Springer},
      YEAR = {1998},
      DOI = {10.1007/978-1-4757-4721-8}

}

@article{Goguen72:MachinesInCategory,
author = {Goguen, J. A.},
title = {{Minimal realization of machines in closed categories}},
volume = {78},
journal = {Bulletin of the American Mathematical Society},
number = {5},
publisher = {American Mathematical Society},
pages = {777 -- 783},
year = {1972},
}

\end{document}